\documentclass[12pt]{article}

\usepackage[hmargin=2.5cm, vmargin={3cm,3cm}, a4paper]{geometry}\usepackage{amsmath,amssymb,amsthm,mathrsfs}
\usepackage{cite}

\usepackage{amsthm}
\usepackage{epsfig,epsf,subfigure,graphicx,graphics}
\usepackage{url,enumerate}
\graphicspath{{fig/}}
\usepackage{float}\usepackage{parskip}
\usepackage{fancyhdr}
\usepackage{xcolor}
\usepackage{physics}
 
\DeclareMathOperator*{\argmin}{argmin} 
\DeclareMathAlphabet\mathrsfso      {U}{rsfso}{m}{n}
\usepackage{amssymb}\usepackage{mathtools}
\usepackage[toc,page]{appendix}
\usepackage{bbm}
\usepackage{centernot}

\usepackage{lettrine}
\usepackage{type1cm}

\usepackage{algorithm}
\usepackage{algpseudocode}
\usepackage{algpseudocode}

\errorcontextlines\maxdimen
\makeatletter
\newcommand*{\algrule}[1][\algorithmicindent]{\makebox[#1][l]{\hspace*{.5em}\vrule height .75\baselineskip depth .25\baselineskip}}
\newcount\ALG@printindent@tempcnta
\def\ALG@printindent{
    \ifnum \theALG@nested>0
        \ifx\ALG@text\ALG@x@notext
            \addvspace{-3pt}
        \else
            \unskip
            \ALG@printindent@tempcnta=1
            \loop
                \algrule[\csname ALG@ind@\the\ALG@printindent@tempcnta\endcsname]
                \advance \ALG@printindent@tempcnta 1
            \ifnum \ALG@printindent@tempcnta<\numexpr\theALG@nested+1\relax
            \repeat
        \fi
    \fi
    }
\usepackage{etoolbox}
\patchcmd{\ALG@doentity}{\noindent\hskip\ALG@tlm}{\ALG@printindent}{}{\errmessage{failed to patch}}
\makeatother

\newtheorem{lemma}{Lemma}
\newtheorem{definition}{Definition}
\newtheorem{proposition}{Proposition}
\newtheorem{theorem}{Theorem}
\newtheorem{corollary}{Corollary}
\newtheorem*{proposition*}{Proposition 3}

\begin{document}

\title{Relation between quantum advantage in supervised learning and quantum computational advantage}

\author{Jordi Pérez-Guijarro$^\dagger$, Alba Pagès-Zamora$^\dagger$, and Javier R. Fonollosa$^\dagger$}

\date{\small $^\dagger$Departament de Teoria del Senyal i Comunicacions,\\ Universitat Politècnica de Catalunya, Barcelona, Spain}

\maketitle

\begin{abstract}
      The widespread use of machine learning has raised the question of quantum supremacy for supervised learning as compared to quantum computational advantage. In fact, a recent work shows that computational and learning advantage are, in general, not equivalent, i.e., the additional information provided by a training set can reduce the hardness of some problems. This paper investigates under which conditions they are found to be equivalent or, at least, highly related. The existence of efficient algorithms to generate training sets emerges as the cornerstone of such conditions. These results are applied to prove that there is a quantum speed-up for some learning tasks based on the prime factorization problem, assuming the classical intractability of this problem.
\end{abstract}

\lettrine[loversize=0.17]{S}{ince} the origin of the quantum computing field, the study of applications that exhibit a quantum speed-up or advantage has played a fundamental role. This can be observed in the seminal works of Feynman\cite{feynman2018simulating} and Deutsch\cite{deutsch1985quantum}, and the subsequent development of quantum algorithms such as Simon's\cite{simon1997power}, Shor's\cite{shor1999polynomial}, and Grover's \cite{grover1996fast}. This endeavor to understand the advantages of quantum computing extends to the present, where there are ongoing efforts \cite{arute2019quantum,zhong2020quantum,madsen2022quantum} to experimentally demonstrate a quantum advantage using noisy intermediate-scale quantum devices through boson sampling \cite{aaronson2011computational} and random circuit sampling \cite{bouland2018quantum}. Additionally, significant attention has been devoted to the design of quantum algorithms for machine learning over the past decade.

The design of these latter algorithms has been explored in various computational settings, utilizing a diverse range of techniques. Some approaches focus on classical data, while others incorporate quantum data \cite{servedio2004equivalences}. In the former case, some results like the HHL algorithm \cite{harrow2009quantum}, which is used in linear systems of equations, and quantum principal component analysis \cite{lloyd2014quantum} seemingly provide an exponential advantage, but the end-to-end advantage is arguable considering the dequantization procedure proposed in \cite{tang2019quantum}. Nonetheless, certain results point to a quantum advantage in supervised learning using mainly two different strategies. The first one improves the optimization process that takes place in supervised learning \cite{rebentrost2014quantum}. This technique usually relies on Grover's algorithm \cite{grover1996fast}, and therefore, the advantage is of quadratic order. In the second approach, a quantum variational circuit replaces the classical model, e.g., a neural network. Indeed, existing theoretical analyses of this perspective \cite{abbas2021power,A_rigorous} suggest that quantum models can outperform classical ones.  Authors in \cite{abbas2021power} conclude that quantum models provide an advantage in terms of trainability and capacity, where the capacity measure is based on the Fisher information matrix. In \cite{A_rigorous}, an exponential quantum advantage for a learning task is proved under the widely accepted assumption that a classical algorithm cannot solve the discrete-logarithm problem efficiently.

Although these encouraging results suggest potential benefits of using quantum algorithms in supervised learning, it remains unclear for which learning tasks these models yield a quantum advantage. Intuitively, one might assume that such advantages arise when attempting to learn functions for which there is a quantum speed-up. However, as shown in \cite{huang2021power} through a counter-example, this statement is not true in general. This highlights the complexity of identifying the specific conditions under which quantum algorithms truly outperform their classical counterparts in supervised learning.

In \cite{gyurik2022establishing}, the authors shed light on certain aspects of quantum advantage in supervised learning. Notably, they identify several examples of learning speed-ups where the classical hardness come from training the model instead of evaluating the resulting model. Furthermore, they present a method that takes a \textbf{BQP-complete} problem and constructs a learning task that exhibits a quantum separation.

In this work we conduct a rigorous analysis of the relationship between quantum advantage in both learning and computational tasks. To achieve this, we depart from the traditional statistical learning framework, typically employed in supervised learning, and adopt a complexity theory perspective. Using this framework, we establish the equivalence between computational and learning speed-ups when the training set can be generated classically. This result is used to prove a quantum advantage in some learning tasks related to the prime factorization problem, under the commonly believed assumption that a classical algorithm cannot solve the prime factorization problem efficiently.

The paper is structured as follows. Section \ref{supervised_learning_and_learning} provides an introduction to the concept of supervised learning, along with the notation used. Next, in Section \ref{quantum_speed_up_in_supervised_learning}, we delve into the topic of quantum speed-up in supervised learning, presenting and discussing two distinct definitions: one associated with a fixed distribution and another independent of the distribution used. Section \ref{sec:realtion_learning_and_computational} explores the relationship between quantum advantage in learning and computational tasks. In this section, we establish that when the training set can be efficiently generated by a classical algorithm, then both concepts are equivalent. This section also shows several examples of learning tasks that exhibit a learning speed-up, which are based on the prime factorization problem. In Section \ref{sec:invariant_functions}, we explore the stronger notion of quantum speed-up introduced in Section \ref{quantum_speed_up_in_supervised_learning}, which is independent of the distribution. To conduct this analysis, we introduce the concept of invariance property. Conclusions are drawn in Section \ref{conclusions}.

\section{Supervised Learning and Learning Procedures}\label{supervised_learning_and_learning}

The problem of supervised learning consists in learning a function $f:\mathcal{X}\rightarrow \mathcal{Y}$ given a training set $\mathcal{T}_n=\{(x_i,y_i)\}_{i=1}^n$ of $n$ sample pairs, where $x_i\in \mathcal{X}$ and $y_i\in \mathcal{Y}$. In our setup, $f(x)$ is a deterministic function, $x_i \in \mathcal{X}$ are realizations of a random variable (r.v.) $X\sim D$, which distribution $D$ is unknown, and the domain is countably infinite, i.e., there exists a bijection from $\mathcal{X}$ to $\mathbb{N}$. Hereafter, we refer to the pair $(f, D)$ as a \emph{learning task}. 

A learning task is solved once a function $\hat{f}:\mathcal{X}\rightarrow \mathcal{Y}$ that approximates $f(x)$ is found for all $x \in \mathrm{supp}(D)$, where $\mathrm{supp}(D)$ denotes the support of $D$. The inferred function $\hat{f}(x)$ belongs to a set of parameterized functions $\mathcal{F}$ called the \emph{hypothesis class}. Specifically, $\mathcal{F}:=\{\mathcal{M}_\theta :\theta \in \mathcal{P}\}$ consists of the functions expressed by a model $\mathcal{M}_\theta$ with parameters $\theta \in \mathcal{P}$, being $\mathcal{P}$ the parameter space. The parameters $\theta\in \mathcal{P}$ are selected through a   learning algorithm, denoted by $\mathcal{A}$. The latter is usually implemented by minimizing a certain cost function on the input training set $\mathcal{T}_n$. A widely used cost function is the \textit{empirical} risk given by $\hat{R}(h)=\frac{1}{n}\sum_{i=1}^n l(h(x_i),y_i)$,  where $l:\mathcal{Y}\times \mathcal{Y}\rightarrow \mathbb{R}^+$ is a loss function and $h \in \mathcal{F}$. That is,
\begin{align*}
\hat{f}=\argmin_{h \in \mathcal{F}}\hat{R}(h)%
\end{align*}
Note that the risk $R(h)=\mathbb{E}_X \left[ l(h(x),y)\right]$ cannot be used as cost function in the optimization process since both $D$ and $f(x)$ are unknown, prompting the utilization of its estimate, the empirical risk. In this paper, we focus instead on a probabilistic criterion to measure the approximation error.

The end-to-end procedure characterized by a learning algorithm and a model, which given input $(x,\mathcal{T}_n)$ outputs an  estimate $\hat{f}(x)$, is referred hereafter as a \emph{learning procedure}.
Note that the Church-Turing thesis guarantees that any algorithm running on a quantum computer can also be implemented by a classical one. Still, their respective complexity might be drastically different making one of them, or even both, unfeasible in practice. For that reason, we restrict our attention to classical or quantum \emph{efficient} learning procedures, defined as mappings $M(x,\mathcal{T}_n)=\hat{f}(x)$ where the running time of the procedure is polynomial in the input size, i.e., $b(\mathcal{T}_n)+b(x)$. The function $b(\cdot)$ yields the number of significant bits required to represent its argument, e.g., $b(z)=\left \lceil  \log_2 z \right \rceil$ if $z$ is a positive integer.

As we assume that $x\in \mathrm{supp}(D)$, the size of $x$ is upper bounded by $m:= \max_{x\in \mathrm{supp}(D)}b(x)$. On the other hand, the size of a
training set with $n$ sample pairs is equal to $b(\mathcal{T}_n)=O(n (m+m^u))=O(n m^{\max\{1,u\}})$, where $f(x)$ is assumed to satisfy $b(f(x))=b(y)\leq b(x)^u$ for a certain $u\in \mathbb{R}^+$ \footnote{The condition $b(f(x))\leq b(x)^u$ is satisfied by any function that can be computed in polynomial time.}. Therefore, the input size of a learning procedure that uses a training set $\mathcal{T}_n$ grows linearly with $n$. The latter implies that an efficient learning procedure runs in polynomial time w.r.t.$\,m$ whenever the number of samples of the training set $n$ is also a polynomial in $m$. Therefore, hereafter we restrict $n$ to be a polynomial in $m$.

\section{Quantum Speed-up in Supervised Learning}\label{quantum_speed_up_in_supervised_learning}

Next, we present the notion of quantum advantage in supervised learning used in this work. A complexity theory point of view is employed, i.e., we focus on the behavior of the number of operations and the error as the input size $m$ grows. 

In particular, we analyze the behavior of learning tasks over sequences of distributions $\mathcal{D}=\{D_m\}_{m=1}^\infty$ such that $\mathrm{supp}(D_m)=\{x\in \mathcal{X}: b(x)\leq m\}$. Since $\mathrm{supp}(D_m)\subset \mathrm{supp}(D_{m+1})$, the learning tasks become more involved as $m$ increases. Note also that different functions may have different domains, and therefore, the distributions are defined on different supports. However, since $\mathcal{X}$ is countably infinite, we can assume without loss of generality that the support of $D_m$ is formed by a subset of binary sequences of at most length $m$, i.e., $\mathrm{supp}(D_m)\subseteq \bigcup_{i\leq m} \{0,1\}^i$.

\begin{definition}\label{def:Learning_classical_quantum}
A function $f(x)$ is classical (quantum) $\eta$-learnable, denoted as $f\in \mathcal{L}_C^{\,\eta}$ $(\mathcal{L}_Q^{\,\eta})$, where $\eta:\mathbb{N}\rightarrow \mathbb{R}^+$, iff there exists a classical (quantum) learning procedure, a sequence of distributions $\mathcal{D}=\{D_m\}_{m=1}^\infty$, and a polynomial $p(m)$ such that:

        Given input $x$ and $\mathcal{T}_{n}=\{(x_i,f(x_i))\}_{i=1}^{p(m)}$, where $x_i$ is distributed according to $D_m$, the classical (quantum) learning procedure computes an estimate $\hat{f}(x)$ in polynomial time in $m$ such that 
    \begin{equation}
        \mathrm{Prob}\left( |\hat{f}(x)-f(x)|<c\,\eta(m)\right)>1-\delta.
    \end{equation}
    for some constants $c\in \mathbb{R}^+$ and $\delta\in[0,\frac{1}{2})$, and for all $x\in \mathrm{supp}(D_m)$. 

\end{definition}

Therefore, if a function is $\eta$-learnable, then it can be estimated with an error $O(\eta(m))$ with at least probability $1/2$ given a training set with a polynomial number of samples. 
Furthermore, note that this definition of learnable functions motivates the following formalization of quantum advantage in supervised learning.

\begin{definition}\label{def:task_specific_speed_up_APZ}
There is a quantum speed-up learning a function $f(x)$ if there exists a function $\eta:\mathbb{N}\rightarrow \mathbb{R}^+$ such that $f(x)$ is quantum $\eta$-learnable but it is not classical $\eta$-learnable, i.e., $f\in \mathcal{L}_Q^{\,\eta}\backslash \mathcal{L}_C^{\,\eta}$.   
\end{definition}

Interestingly, there might be functions for which there is no quantum advantage following Definition $\ref{def:task_specific_speed_up_APZ}$, and still, these functions can be efficiently learned by a quantum procedure but not by a classical one \emph{using sequence} $\mathcal{D}$. More formally,
\begin{equation}
    f\in \mathcal{L}_Q^{\,\eta}(\mathcal{D})\backslash \mathcal{L}_C^{\,\eta}(\mathcal{D})\centernot\implies f\in \mathcal{L}_Q^{\,\eta}\backslash \mathcal{L}_C^{\,\eta} 
\end{equation}
where $\mathcal{L}_Q^{\,\eta}(\mathcal{D})$ $(\mathcal{L}_C^{\,\eta}(\mathcal{D}))$ denotes the set of functions that can be efficiently learned by a classical (quantum) procedure using sequence $\mathcal{D}$. A formal definition of sets $\mathcal{L}_Q^{\,\eta}(\mathcal{D})$ and $\mathcal{L}_C^{\,\eta}(\mathcal{D})$ is presented in Appendix \ref{app:Thm_1}. With a slight abuse of terminology, we say that there is a quantum speed-up learning function $f(x)$ \emph{for a sequence} $\mathcal{D}$ if $f\in \mathcal{L}_Q^{\,\eta}(\mathcal{D})\backslash \mathcal{L}_C^{\,\eta}(\mathcal{D})$.

\section{Relation between learning and computational quantum speed-up}\label{sec:realtion_learning_and_computational}

Now, the relation between quantum advantage in both learning and computational tasks is studied. We start by introducing the definition of (efficiently) computable function used in this work.  

\begin{definition}
    A function is classical (quantum) $\eta$-computable, denoted as $f\in\mathcal{P}_C^{\,\eta}$ $(\mathcal{P}_Q^{\,\eta})$, where $\eta:\mathbb{N}\rightarrow \mathbb{R}^+$, iff there exists a classical (quantum) algorithm that, for any input $x\in \mathcal{X}$, outputs an estimate $\hat{f}(x)$ in polynomial time in $b(x)$ such that 
    \begin{equation}
        \mathrm{Prob}\left( |\hat{f}(x)-f(x)|<c\,\eta(b(x))\right)>1-\delta
    \end{equation}
    for some constants $c\in \mathbb{R}^+$ and $\delta\in[0,\frac{1}{2})$.
\end{definition}

Analogously to the learning scenario, a computational speed-up is defined as follows.

\begin{definition}
    There is a quantum speed-up computing a function $f(x)$ if there exists a function $\eta:\mathbb{N}\rightarrow \mathbb{R}^+$ such that $f(x)$ is quantum $\eta$-computable but it is not classical $\eta$-computable, i.e., $f\in \mathcal{P}_Q^{\,\eta}\backslash \mathcal{P}_C^{\,\eta}$.  
\end{definition}

The work in \cite{huang2021power} shows that quantum computational advantage does not entail quantum advantage in learning tasks, i.e., there are functions that can be computed by a quantum algorithm but not by a classical one, and still, they can be learned by a classical procedure. Specifically, Proposition \ref{proposition_neq} rephrases the result in \cite{huang2021power} in the context of this work. 

\begin{proposition}\label{proposition_neq}
Under the conjecture that the complexity class Bounded-error quantum polynomial time \textnormal{(\textbf{BQP})} is different from the class Bounded-error probabilistic polynomial time \textnormal{(\textbf{BPP})}, the subset of quantum 1-learnable functions for which a quantum speed-up exists differs from the analogous computational set, i.e., $\mathcal{L}_Q^{\,1}\backslash \mathcal{L}_C^{\,1}\neq \mathcal{P}_Q^{\,1}\backslash \mathcal{P}_C^{\,1}$.
\end{proposition}
\begin{proof}
    See Appendix \ref{app:separation}.
\end{proof}
Therefore, an equivalence between $\mathcal{L}_Q^{\,\eta}\backslash \mathcal{L}_C^{\,\eta}$ and $\mathcal{P}_Q^{\,\eta}\backslash \mathcal{P}_C^{\,\eta}$ can not be guaranteed for an arbitrary error trend $\eta:\mathbb{N}\rightarrow \mathbb{R}^+$. This observation motivates the question of whether there is any relationship between learning and computational advantage. Indeed, a relation exists for the functions for which a training set can be efficiently generated. This concept is formally introduced next.   
\begin{definition}
    A function $f(x)$ is classical (quantum) generable for a sequence $\mathcal{D}$, denoted as $f\in \mathcal{G}_C(\mathcal{D})$ $(\mathcal{G}_Q(\mathcal{D}))$, iff there exists a classical (quantum) algorithm that for any input $m$ generates a training set $\mathcal{T}_n=\{(x_i, f(x_i))\}_{i=1}^{n}$ in polynomial time in m with high probability, where $x_i\sim D_m$ and $n$ is a polynomial in $m$.
\end{definition}

As the sets $\mathcal{G}_C(\mathcal{D})$ and $\mathcal{G}_Q(\mathcal{D})$ are defined for a sequence of distributions $\mathcal{D}$, we focus on the corresponding sets $\mathcal{L}_Q^{\,\eta}(\mathcal{D})$ and $\mathcal{L}_C^{\,\eta}(\mathcal{D})$. We are now ready to present the main result of this work that establishes a relation between learning and computational advantage.
    
\begin{theorem}\label{thm:conditions_speed_up}
    A quantum speed-up computing function $f(x)$ is equivalent to a quantum speed-up learning $f(x)$ for a sequence of distributions $\mathcal{D}$ if the function is classical generable for sequence $\mathcal{D}$, i.e.,  
    \begin{equation}\label{eq_conditions_speed_up}
        \left(\mathcal{L}_Q^{\,\eta}(\mathcal{D})\backslash \mathcal{L}_C^{\,\eta}(\mathcal{D})\right)\cap \mathcal{G}_C(\mathcal{D}) = \left(\mathcal{P}_Q^{\,\eta}\backslash \mathcal{P}_C^{\,\eta}\right)\cap \mathcal{G}_C(\mathcal{D})
    \end{equation}
\end{theorem}
\begin{proof}
    See Appendix \ref{app:Thm_1}.
\end{proof}

Under the conjecture that the prime factorization problem is classically intractable, the set $\left(\mathcal{P}_Q^{\,\eta}\backslash \mathcal{P}_C^{\,\eta}\right)\cap \mathcal{G}_C(\mathcal{D})$, and therefore, also $\left(\mathcal{L}_Q^{\,\eta}(\mathcal{D})\backslash \mathcal{L}_C^{\,\eta}(\mathcal{D})\right)\cap \mathcal{G}_C(\mathcal{D})$, differs from the empty set for $\eta(m)=m^u$, and for at least a sequence $\mathcal{D}$ satisfying
        \begin{equation}\label{eq:prop_suff_inf_23}
            \mathrm{Prob}_{X\sim D_m} (x)\geq \frac{\mathrm{Prob}_{X\sim U_m} (x)}{q(m)}
        \end{equation}
 for all $x\in \mathrm{supp}(D_m)$, where $q(m)$ is a polynomial, and $U_m$ denotes the uniform distribution over $\{x\in \mathcal{X}: b(x)\leq m\}$. Specifically, we assume the stronger conjecture that no classical algorithm can  efficiently factorize integers of the form $x=p_1^{r_1}\cdot p_2^{r_2}$. Note that if the latter were false, then the RSA protocol \cite{rivest1978method} could be broken with a classical computer.

 \begin{theorem}\label{thm:examples}
  Let $\mathcal{X}=\{x\in \mathbb{N}:\omega(x)\leq K\}$ be the domain of functions $\{f_l(x)\}_{l=1}^3$, where $\omega(x)$ is the number of distinct prime factors of $x$, which are denoted by $\{p_i\}_{i=1}^{\omega(x)}$, and $K$ is a constant greater than or equal to 2. Under the conjecture that the prime factorization problem is classically intractable for integers with two distinct prime factors:
  
  \begin{enumerate}[(i)]
  \item Function $f_1(x)= \sum_{i=1}^{\omega(x)} p_i$ belongs to $\mathcal{L}_Q^{\,\eta}(\mathcal{D})\backslash\mathcal{L}_C^{\,\eta}(\mathcal{D})$ for all $\eta(m)$ of the form $m^u$, where $u\in\mathbb{R}^+$, as long as $f_1\in \mathcal{G}_C(\mathcal{D})$.
  
  \item At least one of the following functions,
      \begin{equation}
    \begin{matrix} f_2(x)= \prod_{i=1}^{\omega(x)} p_i &       f_3(x)=\left[\sqrt{p_1 \sqrt{p_2\cdots \sqrt{p_{\omega(x)}}}}\right]^4 
     \\
    \end{matrix}
    \end{equation}
    belongs to $\mathcal{L}_Q^{\,\eta}(\mathcal{D})\backslash\mathcal{L}_C^{\,\eta}(\mathcal{D})$ for all $\eta(m)$ of the form $m^u$, as long as $f_2, f_3\in \mathcal{G}_C(\mathcal{D})$.
    
   \end{enumerate}
   Moreover, there exists at least a sequence $\mathcal{D}$ that satisfies \eqref{eq:prop_suff_inf_23} for which $f_l(x)\in \mathcal{G}_C(\mathcal{D})$ for all $l\in\{1,2,3\}$.
\end{theorem}
\begin{proof}
See Appendix \ref{sec:proof_details_prime}.
\end{proof}

\textit{Sketch of the proof}: The first part of the proof uses a reduction argument. In particular, we show that if $f_1(x)$ or both $f_2(x)$ and $f_3(x)$ can be computed efficiently with a polynomial error, i.e., $f_l\in \mathcal{P}_C^{\,\eta}$ for $\eta(m)=m^u$, then there exists a polynomial algorithm that factorizes integers of the form $x=p_1^{r_1}\cdot p_2^{r_2}$ with high probability. Hence, assuming the classical intractability of the prime factorization problem, we conclude that $f_1\notin \mathcal{P}_C^{\,\eta}$ and at least one between $f_2(x)$ and $f_3(x)$ does not belong to $\mathcal{P}_C^{\,\eta}$ for $\eta(m)=m^u$. 

On the other hand, note that the prime factors of $x$, $\{p_i\}_{i=1}^{\omega(x)}$, can be computed in polynomial time with high probability using Shor's algorithm \cite{shor1999polynomial}. Hence, as functions $\{f_l(x)\}_{l=1}^3$ can be computed efficiently once the prime factors are given, then there exists an efficient quantum algorithm that computes these functions, i.e., $f_l\in \mathcal{P}_Q^{\,\eta}$ for $\eta(m)=m^u$. Therefore, using the fact that $f_1\in \mathcal{P}_Q^{\,\eta}\backslash \mathcal{P}_C^{\,\eta}$, and Theorem \ref{thm:conditions_speed_up}, we conclude that $f_1\in \mathcal{L}_Q^{\,\eta}(\mathcal{D})\backslash\mathcal{L}_C^{\,\eta}(\mathcal{D})$, if $f_1\in \mathcal{G}_C(\mathcal{D})$. Analogously, at least one between $f_2(x)$ and $f_3(x)$ belongs to $\mathcal{L}_Q^{\,\eta}(\mathcal{D})\backslash \mathcal{L}_C^{\,\eta}(\mathcal{D})$.

The last part of the proof shows that $f_l\in \mathcal{G}(\mathcal{D})$ for a sequence that satisfies \eqref{eq:prop_suff_inf}. That is, there exists a classical algorithm that generates a training set $\mathcal{T}_n$ for sequence $\mathcal{D}$. The adopted approach generates random samples $(x,\{p_i\}_{i=1}^{\omega(x)})$ and then computes $f_l(x)$ using the prime factors.

 \begin{figure}[H]
	\centering
	\includegraphics[width=7.5cm]{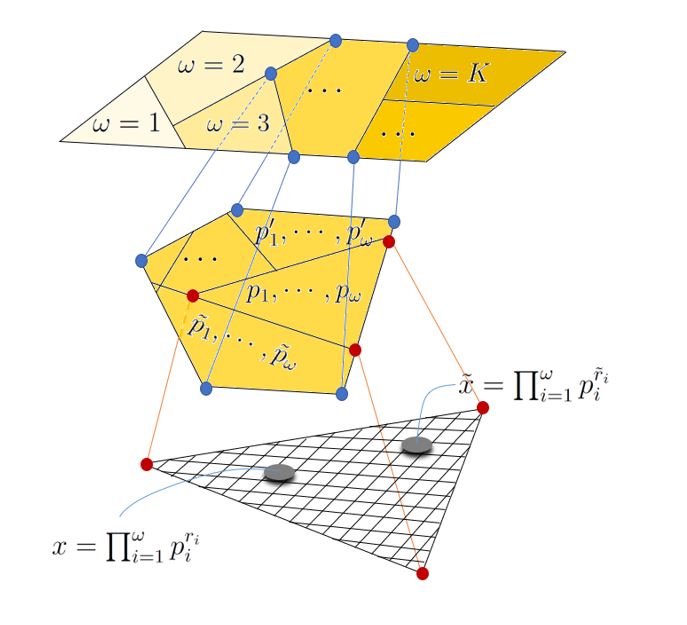}
	\caption{Diagram of the procedure used to generate the training set.}
	\label{fig:generate_samples}
\end{figure}

Figure \ref{fig:generate_samples} represents the three-step procedure followed to generate the random samples $(x,\{p_i\}_{i=1}^{\omega(x)})$ according to a nearly uniform distribution. First, the number of distinct prime numbers $\omega$ is randomly selected, with a probability proportional to the number of integers with $\omega$ distinct prime factors. Once $\omega$ is given, a set of prime numbers $\{p_1,\cdots,p_\omega\}$ is generated randomly. The probability of observing a set $\{p_1,\cdots, p_\omega\}$ is approximately proportional to the number of integers with those prime factors. Finally, the multiplicity $r_i$ of each prime is selected, where each vector $[r_1,\cdots,r_\omega]^T$ that satisfies $\prod_{i=1}^\omega p_i^{r_i}\leq 2^m$ is equally probable.

 \hspace*{15.5cm}  $\square$

On the other hand, Theorem \ref{thm:conditions_speed_up} can be used to characterize the relation between learning and computational advantage when the training set can be generated efficiently by a quantum algorithm, as shown next.

\begin{proposition}\label{prop:advantage_GQ}
    For any error trend $\eta(m)$ and sequence of distributions $\mathcal{D}$, there exists a set $\tilde{\mathcal{G}}(\mathcal{D})$ that satisfies
    \begin{equation}\label{eq_prop_2_1}
        \mathcal{G}_C(\mathcal{D})\subseteq \tilde{\mathcal{G}}(\mathcal{D})\subseteq \mathcal{G}_Q(\mathcal{D})
    \end{equation}
    and, 
    \begin{equation}\label{eq_conditions_speed_up_Q}
        \left(\mathcal{L}_Q^{\,\eta}(\mathcal{D})\backslash \mathcal{L}_C^{\,\eta}(\mathcal{D})\right)\cap \mathcal{G}_Q(\mathcal{D}) = \left(\mathcal{P}_Q^{\,\eta}\backslash \mathcal{P}_C^{\,\eta}\right)\cap \tilde{\mathcal{G}}(\mathcal{D})
    \end{equation}    
\end{proposition}
\begin{proof}
    See Appendix \ref{app:Thm_1}.
\end{proof}

    Similarly to the result of Theorem \ref{thm:conditions_speed_up}, Proposition \ref{prop:advantage_GQ} shows that learning and computational advantage are related when a quantum algorithm can efficiently generate the training set.

 \section{Learning speed-up in invariant functions}\label{sec:invariant_functions}
 
The previous section presents Theorem \ref{thm:conditions_speed_up}, which relates a learning speed-up for a particular sequence of distributions and a computational speed-up. However, this result does not clarify the relation between sets $\mathcal{L}_Q^{\eta}\backslash  \mathcal{L}_C^{\eta}$ and $\mathcal{P}_Q^{\eta}\backslash  \mathcal{P}_C^{\eta}$. To study this relation, we focus on functions for which the error is approximately invariant to small changes in the training set. We refer to this characteristic as the invariance property. 
    
More precisely, we say that a function $f(x)$ is invariant or satisfies the invariance property if the set $\{x\in\mathcal{X}: b(x)\leq m\}$ can be divided into a polynomial number of subsets with similar cardinality, such that when $x$ and $x'$ belong to the same subset $R_j$, then both pairs $(x,f(x))$ and $(x',f(x'))$ are equally informative about function $f(x)$. A formal definition of the invariance property is shown in Appendix \ref{appendix:suff_informative}. 

An important caveat of the formal definition shown in Appendix \ref{appendix:suff_informative} is its dependence on the type of learning procedure employed, whether classical or quantum. To address this, we introduce the classical invariance property and the quantum invariance property as distinct concepts.

    \begin{proposition}\label{prop:suff_inf}
        If $f(x)$ is classical (quantum) $\eta$-learnable and satisfies the classical (quantum) invariance property, then $f(x)$ is classical (quantum) $\eta$-learnable using any sequence of distributions $\mathcal{D}=\{D_m\}_{m=1}^{\infty}$ such that
        \begin{equation}\label{eq:prop_suff_inf}
            \mathrm{Prob}_{X\sim D_m} (x)\geq \frac{\mathrm{Prob}_{X\sim U_m} (x)}{q(m)}
        \end{equation}
        for all $x\in\{x\in \mathcal{X}:b(x)\leq m\}$, where $q(m)$ is a polynomial, and $U_m$ denotes the uniform distribution over $\{x\in \mathcal{X}: b(x)\leq m\}$.
    \end{proposition}
\begin{proof}
See Appendix \ref{appendix:suff_informative}.
\end{proof}

Therefore, if the function to be estimated satisfies the classical invariance property, but it is not classically learnable using a sequence $\mathcal{D}$ that fulfills \eqref{eq:prop_suff_inf}, then the function cannot be classically learned independently of the distribution used. Furthermore, if the function is quantum $\eta$-learnable using sequence $\mathcal{D}$, then $f\in\mathcal{L}_Q^{\eta}\backslash \mathcal{L}_C^{\eta}$, i.e., for this type of functions, we can prove a quantum speed-up as in Definition 2  by demonstrating a quantum advantage for a sequence $\mathcal{D}$ that satisfies \eqref{eq:prop_suff_inf}.

In addition to the aforementioned logical deduction, which is captured by equation (10), we leverage Proposition 3 to derive the following results.

\begin{theorem}\label{thm:non_distribution}
    
    Let $\mathcal{I}_C$ and $\mathcal{I}_Q$ denote the set of classical and quantum invariant functions, and $\mathcal{G}_C$ the set of functions that are classical generable for some sequence $\mathcal{D}$ that satisfies \eqref{eq:prop_suff_inf}, then 
    \begin{equation}\label{classical_invariance_result}
        \left(\mathcal{P}_Q^{\eta}\backslash  \mathcal{P}_C^{\eta}\right) \cap \left(\mathcal{G}_C \cap \mathcal{I}_C \right)\subseteq \left(\mathcal{L}_Q^{\eta}\backslash  \mathcal{L}_C^{\eta}\right) \cap \left(\mathcal{G}_C \cap \mathcal{I}_C \right)
    \end{equation}
    \begin{equation}\label{quantum_invariance_result}
        \left(\mathcal{P}_Q^{\eta}\backslash  \mathcal{P}_C^{\eta}\right) \cap \left(\mathcal{G}_C \cap \mathcal{I}_Q \right)\supseteq \left(\mathcal{L}_Q^{\eta}\backslash  \mathcal{L}_C^{\eta}\right) \cap \left(\mathcal{G}_C \cap \mathcal{I}_Q \right)
    \end{equation}
    which implies that 
    \begin{equation}
        \left(\mathcal{P}_Q^{\eta}\backslash  \mathcal{P}_C^{\eta}\right) \cap \left(\mathcal{G}_C \cap \mathcal{I}_C \cap \mathcal{I}_Q\right)=\left(\mathcal{L}_Q^{\eta}\backslash  \mathcal{L}_C^{\eta}\right) \cap \left(\mathcal{G}_C \cap \mathcal{I}_C\cap \mathcal{I}_Q \right)
    \end{equation}
\end{theorem}
\begin{proof}
See Appendix \ref{proof_theorem_3}.
\end{proof}

Hence, the sets $\mathcal{P}_Q^{\eta}\backslash  \mathcal{P}_C^{\eta}$ and $\mathcal{L}_Q^{\eta}\backslash  \mathcal{L}_C^{\eta}$ exhibit a similar relation to the one shown in \,\,\, Theorem 1.

Furthermore, from Theorem \ref{thm:examples} and Proposition \ref{prop:suff_inf}, the following result is obtained.

\begin{corollary}
 Under the conjecture that the prime factorization problem is classically intractable for integers of the form $x=p_1^{r_1}\cdot p_2^{r_2}$, and assuming that functions $\{f_l(x)\}_{l=1}^3$ satisfy the classical invariance property, then $f_1(x)$ and at least one between $f_2(x)$ and $f_3(x)$ belong to $\mathcal{L}_Q^{\,\eta}\backslash\mathcal{L}_C^{\,\eta}$ for any function $\eta(m)=m^u$, where $u\in \mathbb{R}^+$.
\end{corollary}

Therefore, proving the classical invariance property for these functions implies that they cannot be classically learned independently of the distribution used.

\section{Conclusions}\label{conclusions}

This work contributes to unraveling the relationship between quantum advantage in supervised learning and quantum advantage in computational tasks. With a particular emphasis on the practical scenario where the training set can be efficiently generated by an algorithm, the study delves into this connection from a computational standpoint. By conducting a thorough analysis, it offers valuable insights into the applicability of quantum algorithms for learning tasks, thereby enhancing our understanding of their suitability in practical settings.

Although new insights are provided, still, interesting open questions remain. For instance, it is not clear whether functions $\eta(m)$ different from $\eta(m)=m^u$ satisfy the relation $\mathcal{L}_Q^{\,\eta}\backslash \mathcal{L}_C^{\,\eta}\neq \mathcal{P}_Q^{\,\eta}\backslash \mathcal{P}_C^{\,\eta}$ under the conjecture that \textbf{BQP}$\neq$\textbf{BPP}. Another problem that remains open  is the characterization of the set $\tilde{\mathcal{G}}(\mathcal{D})$, appearing in Proposition 2, and its relation with the sets $\mathcal{G}_C(\mathcal{D})$ and $\mathcal{G}_Q(\mathcal{D})$. As mentioned in Appendix \ref{app:Thm_1}, this characterization requires understanding the relationship between $\mathcal{G}_Q(\mathcal{D})\backslash \mathcal{G}_C(\mathcal{D})$ and $\mathcal{L}_C^{\,\eta}(\mathcal{D})$. A solution to this problem would help to identify learning tasks for which there is a quantum speed-up, and a quantum algorithm can generate the training set efficiently.

\bibliographystyle{ieeetr}
\bibliography{bibliography.bib}

\begin{thebibliography}{10}

\bibitem{feynman2018simulating}
R.~P. Feynman, ``Simulating physics with computers,'' in {\em Feynman and
  computation}, pp.~133--153, CRC Press, 2018.

\bibitem{deutsch1985quantum}
D.~Deutsch, ``Quantum theory, the church--turing principle and the universal
  quantum computer,'' {\em Proceedings of the Royal Society of London. A.
  Mathematical and Physical Sciences}, vol.~400, no.~1818, pp.~97--117, 1985.

\bibitem{simon1997power}
D.~R. Simon, ``On the power of quantum computation,'' {\em SIAM journal on
  computing}, vol.~26, no.~5, pp.~1474--1483, 1997.

\bibitem{shor1999polynomial}
P.~W. Shor, ``Polynomial-time algorithms for prime factorization and discrete
  logarithms on a quantum computer,'' {\em SIAM review}, vol.~41, no.~2,
  pp.~303--332, 1999.

\bibitem{grover1996fast}
L.~K. Grover, ``A fast quantum mechanical algorithm for database search,'' in
  {\em Proceedings of the twenty-eighth annual ACM symposium on Theory of
  computing}, pp.~212--219, 1996.

\bibitem{arute2019quantum}
F.~Arute, K.~Arya, R.~Babbush, D.~Bacon, J.~C. Bardin, R.~Barends, R.~Biswas,
  S.~Boixo, F.~G. Brandao, D.~A. Buell, {\em et~al.}, ``Quantum supremacy using
  a programmable superconducting processor,'' {\em Nature}, vol.~574, no.~7779,
  pp.~505--510, 2019.

\bibitem{zhong2020quantum}
H.-S. Zhong, H.~Wang, Y.-H. Deng, M.-C. Chen, L.-C. Peng, Y.-H. Luo, J.~Qin,
  D.~Wu, X.~Ding, Y.~Hu, {\em et~al.}, ``Quantum computational advantage using
  photons,'' {\em Science}, vol.~370, no.~6523, pp.~1460--1463, 2020.

\bibitem{madsen2022quantum}
L.~S. Madsen, F.~Laudenbach, M.~F. Askarani, F.~Rortais, T.~Vincent, J.~F.
  Bulmer, F.~M. Miatto, L.~Neuhaus, L.~G. Helt, M.~J. Collins, {\em et~al.},
  ``Quantum computational advantage with a programmable photonic processor,''
  {\em Nature}, vol.~606, no.~7912, pp.~75--81, 2022.

\bibitem{aaronson2011computational}
S.~Aaronson and A.~Arkhipov, ``The computational complexity of linear optics,''
  in {\em Proceedings of the forty-third annual ACM symposium on Theory of
  computing}, pp.~333--342, 2011.

\bibitem{bouland2018quantum}
A.~Bouland, B.~Fefferman, C.~Nirkhe, and U.~Vazirani, ``Quantum supremacy and
  the complexity of random circuit sampling,'' {\em arXiv preprint
  arXiv:1803.04402}, 2018.

\bibitem{servedio2004equivalences}
R.~A. Servedio and S.~J. Gortler, ``Equivalences and separations between
  quantum and classical learnability,'' {\em SIAM Journal on Computing},
  vol.~33, no.~5, pp.~1067--1092, 2004.

\bibitem{harrow2009quantum}
A.~W. Harrow, A.~Hassidim, and S.~Lloyd, ``Quantum algorithm for linear systems
  of equations,'' {\em Physical review letters}, vol.~103, no.~15, p.~150502,
  2009.

\bibitem{lloyd2014quantum}
S.~Lloyd, M.~Mohseni, and P.~Rebentrost, ``Quantum principal component
  analysis,'' {\em Nature Physics}, vol.~10, no.~9, pp.~631--633, 2014.

\bibitem{tang2019quantum}
E.~Tang, ``A quantum-inspired classical algorithm for recommendation systems,''
  in {\em Proceedings of the 51st Annual ACM SIGACT Symposium on Theory of
  Computing}, pp.~217--228, 2019.

\bibitem{rebentrost2014quantum}
P.~Rebentrost, M.~Mohseni, and S.~Lloyd, ``Quantum support vector machine for
  big data classification,'' {\em Physical review letters}, vol.~113, no.~13,
  p.~130503, 2014.

\bibitem{abbas2021power}
A.~Abbas, D.~Sutter, C.~Zoufal, A.~Lucchi, A.~Figalli, and S.~Woerner, ``The
  power of quantum neural networks,'' {\em Nature Computational Science},
  vol.~1, no.~6, pp.~403--409, 2021.

\bibitem{A_rigorous}
Y.~Liu, S.~Arunachalam, and K.~Temme, ``A rigorous and robust quantum speed-up
  in supervised machine learning,'' {\em Nature Physics}, vol.~17, July 2021.

\bibitem{huang2021power}
H.-Y. Huang, M.~Broughton, M.~Mohseni, R.~Babbush, S.~Boixo, H.~Neven, J.~R.
  McClean, {\em et~al.}, ``Power of data in quantum machine learning,'' {\em
  Nature Communications}, vol.~12, no.~1, pp.~1--9, 2021.

\bibitem{gyurik2022establishing}
C.~Gyurik and V.~Dunjko, ``Exponential separations between classical and
  quantum learners,'' {\em arXiv preprint arXiv:2306.16028}, 2023.

\bibitem{rivest1978method}
R.~L. Rivest, A.~Shamir, and L.~Adleman, ``A method for obtaining digital
  signatures and public-key cryptosystems,'' {\em Communications of the ACM},
  vol.~21, no.~2, pp.~120--126, 1978.

\bibitem{1629135}
V.~Shende, S.~Bullock, and I.~Markov, ``Synthesis of quantum-logic circuits,''
  {\em IEEE Transactions on Computer-Aided Design of Integrated Circuits and
  Systems}, vol.~25, no.~6, pp.~1000--1010, 2006.

\bibitem{sun2021asymptotically}
X.~Sun, G.~Tian, S.~Yang, P.~Yuan, and S.~Zhang, ``Asymptotically optimal
  circuit depth for quantum state preparation and general unitary synthesis,''
  {\em arXiv preprint arXiv:2108.06150}, 2021.

\bibitem{landau1909handbuch}
E.~Landau, ``Handbuch der lehre vonder verteilung der primzahlen, vol. 1,''
  1909.

\bibitem{bringmann2017efficient}
K.~Bringmann and K.~Panagiotou, ``Efficient sampling methods for discrete
  distributions,'' {\em Algorithmica}, vol.~79, no.~2, pp.~484--508, 2017.

\bibitem{rosser1962approximate}
J.~B. Rosser and L.~Schoenfeld, ``Approximate formulas for some functions of
  prime numbers,'' {\em Illinois Journal of Mathematics}, vol.~6, no.~1,
  pp.~64--94, 1962.

\bibitem{joye2000efficient}
M.~Joye, P.~Paillier, and S.~Vaudenay, ``Efficient generation of prime
  numbers,'' in {\em International Workshop on Cryptographic Hardware and
  Embedded Systems}, pp.~340--354, Springer, 2000.

\end{thebibliography}

\smallskip\smallskip\smallskip

\textbf{Acknowledgements}: This work has been funded by grants PID2019-104958RB-C41 funded by MCIN/AEI/ 10.13039/501100011033 and by grants 2021 SGR 01033 and QuantumCAT 001-P-001644 funded by Departament de Recerca i Universitats de la Generalitat de Catalunya 10.13039/501100002809 and by the European Regional Development Funds (ERDF).

\newpage

\begin{appendices}

\setcounter{equation}{0}
\renewcommand{\theequation}{\thesection.\arabic{equation}}
    
    \section{Proof of Proposition \ref{proposition_neq}}\label{app:separation}

    In this section, we reformulate the results of \cite{huang2021power} about the complexity difference between computational and learning tasks. In particular, we show that $\mathcal{L}_Q^{\,1}\backslash \mathcal{L}_C^{\,1}\neq \mathcal{P}_Q^{\,1}\backslash \mathcal{P}_C^{\,1}$ under the conjecture that \textnormal{\textbf{BQP}}$\neq$\textnormal{\textbf{BPP}}. This is done by constructing a function for which there is a computational speed-up, but there is no learning speed-up.

    To design a function $f(x)$ that satisfies this condition, we use Proposition 1 of \cite{huang2021power}.
    
    \begin{proposition}{\cite{huang2021power}} Consider input vector $x\in \mathbb{R}^p$ encoded into a $k$-qubit state $\ket{x}=\sum_{i=1}^p x_i \ket{i}$. If a randomized classical algorithm can compute
    \begin{equation}
        f(x,U,O)=\bra{x} U^\dagger O U \ket{x}
    \end{equation}
    up to $0.15$-error with high probability over the randomness in the classical algorithm for any $k$, $U$, and $O$ in a time polynomial to the description length of $U$ and $O$, the input vector size $p$, and the qubit system size $k$, then \textnormal{\textbf{BQP}}=\textnormal{\textbf{BPP}}.
    \end{proposition}
    
    Interestingly, from the proof of this proposition, it follows that the same result is satisfied if we fix the observable to be $O=Z\otimes I \otimes \cdots \otimes I$, i.e., only the first qubit is measured. Furthermore, the condition of computing an accurate estimate with high probability can be relaxed to estimate $f(x,U,Z\otimes I \otimes \cdots \otimes I)$ up to 0.15-error with probability $1-\delta>1 /2$.

    Hence, assuming that \textnormal{\textbf{BQP}}$\neq$\textnormal{\textbf{BPP}}, then there exists an unbounded number of unitaries $U$ for which $f(x,U,Z\otimes I\otimes \cdots \otimes I)$ cannot be estimated in polynomial time up to $0.15$-error. Consequently, we can define a sequence $U_k$, such that there is an infinite number of values $k$ for which the circuit cannot be estimated classically up to 0.15-error. Moreover, each of the circuits applies non-trivially on $k$ qubits states, and it is formed by a polynomial in $k$ gates.

    Using this sequence, we design our example, whose domain is $\mathbb{Z}_q^*=\bigcup_{\ell \geq 1} \{0,\cdots,q-1\}^\ell$. Therefore, $b(x)=\left \lceil \log q \right \rceil  \cdot \ell(x)$, where $\ell(x)$ denotes the length of string $x$. In particular, the function is defined as
    \begin{equation}\label{example_function}
        f(x)= \ell(x) \bra{x} U_{\ell(x)}^\dagger\left( Z\otimes I\otimes \cdots \otimes I \right) U_{\ell(x)} \ket{x}
    \end{equation}
    where $\ket{x}=\frac{1}{\left\|x \right\|_2}\sum_{i=1}^{\ell(x)}x_i\ket{i}$. Next, we prove that this function yields a computational speed-up but not a learning speed-up.  

        \subsection{$f(x)$ is quantum 1-computable}

            Measurements from state $U_{\ell(x)}\ket{x}$ are performed to compute $f(x)$ with a quantum algorithm. This state can be generated using $\mathrm{poly}(\ell(x))$ gates, where $\mathrm{poly}(\ell(x))$ indicates a term at
            most polynomial in $\ell(x)$. This follows since the encoding needs $O(\ell(x))$ gates \cite{1629135,sun2021asymptotically} and the unitary $U_{\ell(x)}$ requires $\mathrm{poly}(\ell(x))$. Note that the empirical mean deviates from the expected value in at most $\frac{c}{\ell(x)}$ with  probability $1-\delta$ if we use
            \begin{equation}
                O\left(\ell(x)^2\frac{\ln \left( \frac{1}{\delta}\right)}{c^2}\right) 
            \end{equation} 
            samples. The latter follows from Hoeffding's inequality. Therefore, $f(x)$ can be estimated with an error $c$ with probability $1-\delta$ in a running time 
            \begin{equation}
                O\left( \mathrm{poly}(b(x))\frac{\ln \left( \frac{1}{\delta}\right)}{c^2}\right).
            \end{equation}

        \subsection{ $f(x)$ is not classical 1-computable}

            If there exists a classical algorithm that for any input $x$ outputs in polynomial time in $b(x)$ an estimate $\hat{f}(x)$ such that
            \begin{equation}
                \mathrm{Prob}(|\hat{f}(x)-f(x)|\leq c)>\frac{1}{2},
            \end{equation}
            then the ratio 
            \begin{equation}
                \frac{\hat{f}(x)}{\ell(x)}=  \bra{x} U_{\ell(x)}^\dagger \left( Z\otimes I\otimes \cdots \otimes I \right) U_{\ell(x)} \ket{x}+ \frac{c}{\ell(x)}
            \end{equation}
            can be efficiently computed. Hence, for a sufficiently large $\ell(x)$, the function $f(x,U_{\ell(x)},Z\otimes I\otimes \cdots \otimes I)$ can be computed in polynomial time with an error smaller than 0.15, which contradicts the hardness assumption.

        \subsection{ $f(x)$ is classical 1-learnable}

    To prove that $f(x)$ is classical 1-learnable, we use that $f(x)$ can be expressed as
    \begin{equation}
        f(x)=\ell(x) \sum_{i,j=1}^{\ell(x)} x_i x_j M^{(\ell(x))}_{i,j}
    \end{equation}
    where $M^{(k)}_{i,j}= \left[U_k^\dagger \left( Z\otimes I\otimes \cdots \otimes I\right) U_k \right]_{i,j}$. As $M^{(k)}$ is an Hermitian matrix, 
    \begin{equation}
        f(x)=\ell(x)\left(\sum_{1\leq i < j\leq \ell(x)} 2 \mathrm{Re}\{M^{(\ell(x))}_{i,j}\}  x_i x_j + \sum_{i=1}^{\ell(x)} x_i^2 M^{(\ell(x))}_{i,i} \right)
    \end{equation}
    Defining the vectors $v(x):=[x_1^2,\cdots, x_{\ell(x)}^2, 2x_1x_2, \cdots, 2x_{\ell(x)-1}x_{\ell(x)} ]^T$ and 
    \begin{equation}
       v_M(\ell(x)):=[M^{(\ell(x))}_{1,1},\cdots, M^{(\ell(x))}_{\ell(x),\ell(x)}, \mathrm{Re}\{M^{(\ell(x))}_{1,2}\},\cdots,  \mathrm{Re}\{M^{(\ell(x))}_{\ell(x)-1,\ell(x)}\}]^T 
    \end{equation}
    the function becomes
    \begin{equation}
        f(x)= \ell(x) v(x)^T v_M(\ell(x))
    \end{equation}
    Therefore, a training set $\{(x_i,f(x_i))\}_{i=1}^n$, whose samples satisfy $\ell(x_i)=\ell$ for all $i\in[1:n]$, can be expressed as 
    \begin{equation}
        \begin{bmatrix}
        f(x_1)\\ 
        f(x_2)\\ 
        \cdots\\
        f(x_n)
        \end{bmatrix}=    \ell    \begin{bmatrix}
        v(x_1)^T\\ 
        v(x_2)^T\\ 
        \cdots\\
        v(x_n)^T
        \end{bmatrix} v_M(\ell)
    \end{equation}
    or equivalently, as $v_f= \ell \, V_x v_M(\ell)$, where $v_f\in \mathbb{R}^n$, $V_x\in \mathbb{R}^{n\times\frac{\ell(\ell+1)}{2}}$, and $v_M(\ell)\in \mathbb{R}^{\frac{\ell(\ell+1)}{2}}$. This implies that $v_M(\ell)$ can be computed solving a least squares problem, i.e., $v_M(\ell)=\frac{1}{\ell}(V_x^TV_x)^{-1}V_x^T v_f$ if the training set contains $n>\frac{\ell(\ell+1)}{2}$ uniformly distributed samples.
    
    In conclusion, we can estimate classically the function for any value $x$ such that $\ell(x)=\ell$ if a sufficiently large training set is given. Note that to estimate $f(x)$ for all $x$ such that $\ell(x)\leq \ell$, it is necessary $O(\ell^3)$ samples, this follows from 
    \begin{equation}
        \sum_{i=1}^\ell i^2= \frac{\ell^3}{3}+\frac{\ell^2}{2}+\frac{\ell}{6}
    \end{equation}

    \subsection{Generalization to polynomial trends}

    The result introduced in Proposition \ref{proposition_neq} can be generalized with an analogous argument for the case $\eta(m)=m^{u}$ with $u\in \mathbb{R}^+$. In particular, this generalization is based on function  
    \begin{equation}\label{example_function_generalization}
        f(x)= \ell(x)^{u+1} \bra{x} U_{\ell(x)}^\dagger\left( Z\otimes I\otimes \cdots \otimes I \right) U_{\ell(x)} \ket{x},
    \end{equation}
    which belongs to $\mathcal{P}_Q^{\,\eta} \backslash \mathcal{P}_C^{\,\eta}$ and $\mathcal{L}_C^{\,\eta}$ for $\eta(m)=m^u$. This two results follow from: $(i)$ $\bra{x} U_{\ell(x)}^\dagger\left( Z\otimes I\otimes \cdots \otimes I \right) U_{\ell(x)} \ket{x}$ can be estimated with an additive error $c/ \ell(x)^{u+1}$ in polynomial time in $b(x)$ with a quantum algorithm; $(ii)$ Similarly to the previous example, $f\notin \mathcal{P}_C^{\,\eta}$, as otherwise the ratio $f(x)/\ell(x)^{u+1}$ can be compute classically in polynomial time with an error $c/\ell(x)$ producing a contradiction; $(iii)$ $f(x)$ can be learned using the same procedure used previously.

    \section{Proofs of Theorem \ref{thm:conditions_speed_up} and Proposition \ref{prop:advantage_GQ}}\label{app:Thm_1}

    \setcounter{equation}{0}

    In this section the proofs of Theorem \ref{thm:conditions_speed_up} and Proposition \ref{prop:advantage_GQ} are shown, but first, we introduce the formal definition of sets $\mathcal{L}_C^\eta(\mathcal{D})$ and $\mathcal{L}_Q^\eta(\mathcal{D})$.
    
\begin{definition}\label{def:learning_with_distribution}
    A function $f(x)$ is classical (quantum) $\eta$-learnable for a sequence $\mathcal{D}$, denoted as $f\in \mathcal{L}_C^{\,\eta}(\mathcal{D})$ $(\mathcal{L}_Q^{\,\eta}(\mathcal{D}))$, where $\eta:\mathbb{N}\rightarrow \mathbb{R}^+$, and $\mathrm{supp}(D_m)=\bigcup_{i\leq m} \{0,1\}^i$, iff there exists both a classical (quantum) learning procedure and a polynomial $p(m)$ such that:

    Given input $x$ and $\mathcal{T}_{n}=\{(x_i,f(x_i))\}_{i=1}^{p(m)}$, where $x_i$ is distributed according to
    \begin{equation}\label{restricted_distribution}
        \mathrm{Prob}_{X\sim \tilde{D}_m}(x)=\frac{ \mathrm{Prob}_{X\sim D_m} (x)}{\sum_{x'\in \{x\in \mathcal{X}:\,b(x)\leq m\}}\mathrm{Prob}_{X\sim D_m} (x')},
    \end{equation}
    the classical (quantum) learning procedure computes an estimate $\hat{f}(x)$ in polynomial time in $m$ such that 
    \begin{equation}
        \mathrm{Prob}\left( |\hat{f}(x)-f(x)|<c\,\eta(m)\right)>1-\delta.
    \end{equation}
    for some constants $c\in \mathbb{R}^+$ and $\delta\in[0,\frac{1}{2})$, and for all $x\in \mathcal{X}$ such that $b(x)\leq m$. 
\end{definition}
    
Similarly, in the definition of sets $\mathcal{G}_C(\mathcal{D})$ and $\mathcal{G}_Q(\mathcal{D})$, $\mathcal{D}$ must be interpreted in the same way as in the previous definition. Next, the auxiliary results of Lemma \ref{complexity_identities} are demonstrated.

\begin{lemma}\label{complexity_identities}
The following identities hold:
    \begin{align*}
        &(a)\,\, \mathcal{P}_C^{\,\eta} \subseteq \mathcal{L}_C^{\,\eta}(\mathcal{D}) \;\;\;\;\;\;\;\;\;\;\;\;\;\;\;\;\,(b)\,\, \mathcal{P}_Q^{\,\eta} \subseteq \mathcal{L}_Q^{\,\eta}(\mathcal{D})  \\&(c)\,\,
         \mathcal{G}_C(\mathcal{D})\cap \mathcal{L}_C^{\,\eta}(\mathcal{D})	\subseteq \mathcal{P}_C^{\,\eta}\nonumber \;\;\;(d)\,\, \mathcal{G}_Q(\mathcal{D})\cap \mathcal{L}_Q^{\,\eta}(\mathcal{D})	\subseteq \mathcal{P}_Q^{\,\eta}\nonumber \\&(e)\,\,
        \mathcal{G}_C(\mathcal{D})\cap \mathcal{L}_Q^{\,\eta}(\mathcal{D})= \mathcal{G}_C(\mathcal{D})\cap \mathcal{P}_Q^{\,\eta}
    \end{align*}
\end{lemma}

\begin{proof}
    The identities $\mathcal{P}_C^{\,\eta} \subseteq \mathcal{L}_C^{\,\eta}(\mathcal{D})$ and $\mathcal{P}_Q^{\,\eta} \subseteq \mathcal{L}_Q^{\,\eta}(\mathcal{D})$ follow trivially as the learning procedure introduced in Definition \ref{def:learning_with_distribution} may not use the additional information given by the training set $\mathcal{T}_n$.

     As for identity $(c)$, note that if $f\in 
     \mathcal{L}_C^{\,\eta}(\mathcal{D}) \cap \mathcal{G}_C(\mathcal{D})$, then there exists an algorithm that runs in polynomial time that computes $f(x)$. The algorithm on input $x$ generates a training set consisting of $n=p(b(x))$ pairs, where $p(b)$ is a polynomial. The latter can be done efficiently since $f\in \mathcal{G}_C(\mathcal{D})$. Hereafter, a learning procedure on input $x$, and $\mathcal{T}_n$ is used to obtain an estimate $\hat{f}(x)$, which can also be done efficiently since $f\in \mathcal{L}_C^{\,\eta}(\mathcal{D})$. Therefore, $f\in \mathcal{P}_C^{\,\eta}$, or equivalently, $\mathcal{L}_C^{\,\eta}(\mathcal{D}) \cap \mathcal{G}_C(\mathcal{D})\subseteq \mathcal{P}_C^{\,\eta}$. The proof of identity $(d)$ is analogous.

    Finally, we prove identity $(e)$,
    \begin{align}
        \mathcal{L}_Q^{\,\eta}(\mathcal{D}) \cap \mathcal{G}_C(\mathcal{D})&= \mathcal{G}_C(\mathcal{D}) \cap \left( \mathcal{P}_Q^{\,\eta}\cup(\mathcal{L}_Q^{\,\eta}(\mathcal{D})\backslash\mathcal{P}_Q^{\,\eta}) \right)\nonumber\\&=\left(\mathcal{G}_C(\mathcal{D})\cap \mathcal{P}_Q^{\,\eta} \right)\cup \left( \mathcal{G}_C(\mathcal{D})\cap \left(\mathcal{L}_Q^{\,\eta}(\mathcal{D})\backslash \mathcal{P}_Q^{\,\eta}\right) \right)\nonumber\\&=\mathcal{G}_C(\mathcal{D})\cap \mathcal{P}_Q^{\,\eta} 
    \end{align}
    
    where the first step follows from the fact that if $A\subseteq B$ then $B=A\cup (B\backslash A)$, and the third step uses $(\mathcal{L}_Q^{\,\eta}(\mathcal{D})\backslash \mathcal{P}_Q^{\,\eta})\cap \mathcal{G}_C(\mathcal{D})=\emptyset$. To prove this last identity, note that using $\mathcal{G}_C(\mathcal{D}) \subseteq \mathcal{G}_Q(\mathcal{D})$, we obtain that 
    \begin{equation}
        \mathcal{L}_Q^{\,\eta}(\mathcal{D})\cap \mathcal{G}_C(\mathcal{D})\subseteq	\mathcal{L}_Q^{\,\eta}(\mathcal{D})\cap \mathcal{G}_Q(\mathcal{D})\subseteq\mathcal{P}_Q^{\,\eta}
    \end{equation}
    This means that if $f\in \mathcal{L}_Q^{\,\eta}(\mathcal{D})\backslash \mathcal{P}_Q^{\,\eta}$ then $f\notin \mathcal{L}_Q^{\,\eta}(\mathcal{D})\cap \, \mathcal{G}_C(\mathcal{D})$, or equivalently, $(\mathcal{L}_Q^{\,\eta}(\mathcal{D})\backslash \mathcal{P}_Q^{\,\eta})\cap \mathcal{G}_C(\mathcal{D})=\emptyset$.
\end{proof}
    
    \subsection{Proof of Theorem \ref{thm:conditions_speed_up}}
Now, using the results of the previous lemma Theorem \ref{thm:conditions_speed_up} is proved, i.e.,
\begin{equation}
    \left(\mathcal{L}_Q^{\,\eta}(\mathcal{D})\backslash \mathcal{L}_C^{\,\eta}(\mathcal{D})\right)\cap \mathcal{G}_C(\mathcal{D}) = \left(\mathcal{P}_Q^{\,\eta}\backslash \mathcal{P}_C^{\,\eta}\right)\cap \mathcal{G}_C(\mathcal{D})
\end{equation}
We use that if $A\subseteq B$ and $B\subseteq A$ then $A=B$. First, we show that $(\mathcal{L}_Q^{\,\eta}(\mathcal{D})\backslash \mathcal{L}_C^{\,\eta}(\mathcal{D}))\cap \mathcal{G}_C(\mathcal{D})\supseteq (\mathcal{P}_Q^{\,\eta}\backslash \mathcal{P}_C^{\,\eta})\cap \mathcal{G}_C(\mathcal{D})$. 
    \begin{align*}
        \left(\mathcal{L}_Q^{\,\eta}(\mathcal{D})\backslash\mathcal{L}_C^{\,\eta}(\mathcal{D})\right) \cap \mathcal{G}_C(\mathcal{D}) &= \left(\mathcal{L}_Q^{\,\eta}(\mathcal{D}) \cap \mathcal{G}_C(\mathcal{D})\right) \backslash \left(\mathcal{L}_C^{\,\eta}(\mathcal{D}) \cap \mathcal{G}_C(\mathcal{D})\right)\nonumber \\& \supseteq \left(\mathcal{L}_Q^{\,\eta}(\mathcal{D}) \cap \mathcal{G}_C(\mathcal{D})\right) \backslash \mathcal{P}_C^{\,\eta} \supseteq \mathcal{P}_Q^{\,\eta} \cap \mathcal{G}_C(\mathcal{D}) \cap \overline{\mathcal{P}_C^{\,\eta}}
    \end{align*}
    In the first inequality is used that $\mathcal{G}_C(\mathcal{D})\cap \mathcal{L}_C^{\,\eta}(\mathcal{D})	\subseteq \mathcal{P}_C^{\,\eta}$, and the second inequality follows from $\mathcal{P}_Q^{\,\eta} \subseteq \mathcal{L}_Q^{\,\eta}(\mathcal{D})$. Similarly,
    \begin{align*}
        \left(\mathcal{L}_Q^{\,\eta}(\mathcal{D})\backslash \mathcal{L}_C^{\,\eta}(\mathcal{D})\right)\cap \mathcal{G}_C(\mathcal{D})&= \left(\mathcal{L}_Q^{\,\eta}(\mathcal{D}) \cap \mathcal{G}_C(\mathcal{D})\right) \backslash \left(\mathcal{L}_C^{\,\eta}(\mathcal{D}) \cap \mathcal{G}_C(\mathcal{D})\right)\\&=\left(\mathcal{P}_Q^{\,\eta} \cap \mathcal{G}_C(\mathcal{D})\right)\backslash \left(\mathcal{L}_C^{\,\eta}(\mathcal{D}) \cap \mathcal{G}_C(\mathcal{D})\right)\nonumber \\&=\left(\mathcal{P}_Q^{\,\eta} \backslash \mathcal{L}_C^{\,\eta}(\mathcal{D})\right)\cap \mathcal{G}_C(\mathcal{D})\subseteq \left(\mathcal{P}_Q^{\,\eta} \backslash \mathcal{P}_C^{\,\eta}\right)\cap \mathcal{G}_C(\mathcal{D})
    \end{align*}
    
    The second equality is obtained using $\mathcal{G}_C(\mathcal{D})\cap \mathcal{L}_Q^{\,\eta}(\mathcal{D})= \mathcal{G}_C(\mathcal{D})\cap \mathcal{P}_Q^{\,\eta}$, and the inequality follows from $\mathcal{P}_C^{\,\eta}\subseteq \mathcal{L}_C^{\,\eta}(\mathcal{D})$. Hence, 
    \begin{equation}
        \left(\mathcal{L}_Q^{\,\eta}(\mathcal{D})\backslash \mathcal{L}_C^{\,\eta}(\mathcal{D})\right)\cap \mathcal{G}_C(\mathcal{D})=\left(\mathcal{P}_Q^{\,\eta} \backslash \mathcal{P}_C^{\,\eta}\right)\cap \mathcal{G}_C(\mathcal{D}).
    \end{equation}

    \subsection{Proof of Proposition \ref{prop:advantage_GQ}}

    First, note that Proposition $\ref{prop:advantage_GQ}$ is equivalent to
    \begin{equation}\label{first_inclusion_prop}
        \left(\mathcal{P}_Q^{\,\eta}\backslash \mathcal{P}_C^{\,\eta}\right)\cap \mathcal{G}_C(\mathcal{D})\subseteq \left(\mathcal{L}_Q^{\,\eta}(\mathcal{D})\backslash \mathcal{L}_C^{\,\eta}(\mathcal{D})\right)\cap \mathcal{G}_Q(\mathcal{D})
    \end{equation}
    and 
    \begin{equation}
        \left(\mathcal{P}_Q^{\,\eta}\backslash \mathcal{P}_C^{\,\eta}\right)\cap \mathcal{G}_Q(\mathcal{D})\supseteq \left(\mathcal{L}_Q^{\,\eta}(\mathcal{D})\backslash \mathcal{L}_C^{\,\eta}(\mathcal{D})\right)\cap \mathcal{G}_Q(\mathcal{D})
    \end{equation}
    The first inclusion follows using $\mathcal{G}_C(\mathcal{D})\subseteq \mathcal{G}_Q(\mathcal{D})$, and Theorem \ref{thm:conditions_speed_up}.
    On the other hand, 
    \begin{align}
        \left(\mathcal{L}_Q^{\,\eta}(\mathcal{D})\backslash\mathcal{L}_C^{\,\eta}(\mathcal{D})\right) \cap \mathcal{G}_Q(\mathcal{D}) &= \left(\left(\mathcal{L}_Q^{\,\eta}(\mathcal{D}) \cap \mathcal{G}_Q(\mathcal{D})\right) \cap \overline{\mathcal{L}_C^{\,\eta}(\mathcal{D}})\right)\cap \mathcal{G}_Q(\mathcal{D}) \nonumber \\&\subseteq \left(\mathcal{P}_Q^{\,\eta} \cap \overline{\mathcal{L}_C^{\,\eta}(\mathcal{D}}) \right) \cap \mathcal{G}_Q(\mathcal{D})\nonumber \\&\subseteq  \left(\mathcal{P}_Q^{\,\eta} \cap \overline{\mathcal{P}_C^{\,\eta}}\right)\cap \mathcal{G}_Q(\mathcal{D})
    \end{align}
where the first equality follows from $A\cap A = A$, the first inclusion uses that $\mathcal{L}_Q^{\,\eta}(\mathcal{D}) \cap \mathcal{G}_Q(\mathcal{D}) \subseteq \mathcal{P}_Q^{\,\eta}$, and for the second inclusion is used that $\mathcal{P}_C^{\,\eta} \subseteq \mathcal{L}_C^{\,\eta}(\mathcal{D})$.

\subsection{Characterization of $\tilde{\mathcal{G}}(\mathcal{D})$}

The previous proof shows the existence of a set $\tilde{\mathcal{G}}(\mathcal{D})$ satisfying the identities \eqref{eq_prop_2_1} and \eqref{eq_conditions_speed_up_Q}, but it does not characterize this set. Next, we characterize it.

First, as 
\begin{equation*}
    \left(\mathcal{L}_Q^{\,\eta}(\mathcal{D})\backslash\mathcal{L}_C^{\,\eta}(\mathcal{D})\right)\cap \mathcal{G}_Q(\mathcal{D})=\left(\mathcal{P}_Q^{\,\eta}\backslash\mathcal{P}_C^{\,\eta}\right)\cap \tilde{\mathcal{G}}(\mathcal{D})
\end{equation*}
and $\tilde{\mathcal{G}}(\mathcal{D})\subseteq \mathcal{G}_Q(\mathcal{D})$, then
\begin{align}
    \left(\mathcal{P}_Q^{\,\eta}\backslash\mathcal{P}_C^{\,\eta}\right)\cap \tilde{\mathcal{G}}(\mathcal{D})&=\left(\left(\mathcal{L}_Q^{\,\eta}(\mathcal{D})\backslash\mathcal{L}_C^{\,\eta}(\mathcal{D})\right)\cap   \mathcal{G}_Q(\mathcal{D})\right) \cap \left(\left(\mathcal{P}_Q^{\,\eta}\backslash\mathcal{P}_C^{\,\eta}\right)\cap   \mathcal{G}_Q(\mathcal{D})\right)\nonumber\\&= \left(\mathcal{P}_Q^{\,\eta}\backslash\mathcal{P}_C^{\,\eta}\right) \cap \left(\mathcal{L}_Q^{\,\eta}(\mathcal{D})\backslash\mathcal{L}_C^{\,\eta}(\mathcal{D})\right) \cap \mathcal{G}_Q(\mathcal{D})
\end{align}
As $\mathcal{G}_C(\mathcal{D})\subseteq \mathcal{G}_Q(\mathcal{D})$,  
\begin{equation}
    \left(\mathcal{P}_Q^{\,\eta}\backslash\mathcal{P}_C^{\,\eta}\right) \cap \left(\mathcal{L}_Q^{\,\eta}(\mathcal{D})\backslash\mathcal{L}_C^{\,\eta}(\mathcal{D})\right) \cap \left(\mathcal{G}_Q(\mathcal{D}) \backslash \mathcal{G}_C(\mathcal{D}) \cup \mathcal{G}_C(\mathcal{D})\right)
\end{equation}
and equivalently,
\begin{equation}
    \left(\mathcal{P}_Q^{\,\eta}\backslash\mathcal{P}_C^{\,\eta}\right) \cap \left(\left[\mathcal{L}_Q^{\,\eta}(\mathcal{D})\backslash\mathcal{L}_C^{\,\eta}(\mathcal{D})\cap \mathcal{G}_C(\mathcal{D})\right] \cup \left[ \mathcal{L}_Q^{\,\eta}(\mathcal{D})\backslash\mathcal{L}_C^{\,\eta}(\mathcal{D}) \cap \mathcal{G}_Q(\mathcal{D}) \backslash \mathcal{G}_C(\mathcal{D}) \right]\right)
\end{equation}
Using Theorem \ref{thm:conditions_speed_up}, and some trivial identities, the previous expression becomes, 
\begin{align}
       &\left(\mathcal{P}_Q^{\,\eta}\backslash\mathcal{P}_C^{\,\eta}\right) \cap \left(\left[\mathcal{P}_Q^{\,\eta}\backslash\mathcal{P}_C^{\,\eta}\cap \mathcal{G}_C(\mathcal{D})\right] \cup \left[ \mathcal{L}_Q^{\,\eta}(\mathcal{D})\backslash\mathcal{L}_C^{\,\eta}(\mathcal{D}) \cap \mathcal{G}_Q(\mathcal{D}) \backslash \mathcal{G}_C(\mathcal{D}) \right]\right)
       \nonumber\\&=\left(\mathcal{P}_Q^{\,\eta}\backslash\mathcal{P}_C^{\,\eta}\cap \mathcal{G}_C(\mathcal{D})\right) \cup \left(\mathcal{P}_Q^{\,\eta}\backslash\mathcal{P}_C^{\,\eta}\cap \mathcal{L}_Q^{\,\eta}(\mathcal{D})\backslash\mathcal{L}_C^{\,\eta}(\mathcal{D}) \cap \mathcal{G}_Q(\mathcal{D}) \backslash \mathcal{G}_C(\mathcal{D}) \right)\nonumber\\&=       \left(\mathcal{P}_Q^{\,\eta}\backslash\mathcal{P}_C^{\,\eta}\cap \mathcal{G}_C(\mathcal{D})\right) \cup \left(\mathcal{P}_Q^{\,\eta}\cap\overline{\mathcal{P}_C^{\,\eta}}\cap \overline{\mathcal{L}_C^{\,\eta}(\mathcal{D}) }\cap \mathcal{G}_Q(\mathcal{D}) \backslash \mathcal{G}_C(\mathcal{D}) \right)\nonumber\\&=
       \left(\mathcal{P}_Q^{\,\eta}\backslash\mathcal{P}_C^{\,\eta}\right) \cap \left(  \mathcal{G}_C(\mathcal{D}) \cup \left( \mathcal{G}_Q(\mathcal{D}) \backslash \mathcal{G}_C(\mathcal{D}) \cap \overline{\mathcal{L}_C^{\,\eta}(\mathcal{D}) }\right)\right)
\end{align}
Therefore, $\tilde{\mathcal{G}}(\mathcal{D})=\mathcal{G}_C(\mathcal{D}) \cup \left(  \mathcal{G}_Q(\mathcal{D}) \backslash \mathcal{G}_C(\mathcal{D}) \cap \overline{\mathcal{L}_C^{\,\eta} (\mathcal{D})} \right)$. Hence, in case that $\mathcal{G}_Q(\mathcal{D}) \backslash \mathcal{G}_C(\mathcal{D}) \cap \overline{\mathcal{L}_C^{\,\eta}(\mathcal{D}) }=\mathcal{G}_Q(\mathcal{D}) \backslash \mathcal{G}_C(\mathcal{D})$, then $\tilde{\mathcal{G}}(\mathcal{D})$ is equal to $\mathcal{G}_Q(\mathcal{D})$. Similarly, if $\mathcal{G}_Q(\mathcal{D}) \backslash \mathcal{G}_C(\mathcal{D}) \cap \overline{\mathcal{L}_C^{\,\eta}(\mathcal{D}) }=\emptyset$ then $\tilde{\mathcal{G}}(\mathcal{D})=\mathcal{G}_C(\mathcal{D})$.

    \section{Proof of Proposition \ref{prop:suff_inf}}\label{appendix:suff_informative}

    \setcounter{equation}{0}

    This section shows the proof of Proposition \ref{prop:suff_inf}. The proof is based on the invariance property and Lemma \ref{lemma:increase_prob}. First, this property is introduced formally.

        \begin{definition}\label{def:invariance_property}
    A function $f(x)$ satisfies the classical (quantum) invariance property if there exists a cover $\{R_j\}_{j=1}^{r(m)}$ of $\{x\in \mathcal{X}: b(x)\leq m\}$ formed by a polynomial in $m$ number of subsets $R_j$ of the same cardinality, such that 
    
    \begin{enumerate}[(i)]
            \item for any pair of training sets $\mathcal{T}_n(D_m)$ and $\mathcal{T}_{n'}(D_m')$ for which
        \begin{equation*}\label{eq_inv_res}
             \mathrm{Prob}_{\mathcal{T}_n(D_m),\mathcal{T}_{n'}(D_m')} \left( N_{R_j}(\mathcal{T}_n(D_m))\leq N_{R_j}(\mathcal{T}_{n'}(D_m')) \text{ for all }j\in[1: r(m)]\right)\geq 1-\phi
        \end{equation*}
        where $\phi\in(0,1)$ and $N_{R_j}(\mathcal{T}_n(D_m))$ denotes the number of distinct elements of $\mathcal{T}_n(D_m)$ with a value $x\in R_j$,
        \item and for any classical (quantum) efficient learning procedure $M(x,\mathcal{T}_n(D_m))=\hat{f}(x)$,
    \end{enumerate}
    
       there exists a classical (quantum) efficient learning procedure $M'(x,\mathcal{T}_{n'}(D_m'))=\hat{f}'(x)$ such that 
    \begin{equation}
        \mathrm{Prob}_{\mathcal{T}_{n}(D_m),\mathcal{T}_{n'}(D_m')}\left(|\hat{f}'(x)-f(x)|\leq |\hat{f}(x)-f(x)|\right) \geq 1-  \phi
    \end{equation}
    for all $x\in\{x\in \mathcal{X}:b(x)\leq m\}$.
    \end{definition}

    Next, we prove the following auxiliary result.
    
    \begin{lemma}\label{lemma:increase_prob}
    If estimates $\hat{f}(x)=M(x,\mathcal{T}_{n}(D_m))$ and $\hat{f}'(x)=M'(x,\mathcal{T}_{n'}(D_m'))$ satisfy
    \begin{equation}
        \mathrm{Prob}_{\mathcal{T}_{n}(D_m),\mathcal{T}_{n'}(D_m')}\left(|\hat{f}'(x)-f(x)|\leq |\hat{f}(x)-f(x)|\right) \geq 1-  \phi
    \end{equation}
    then 
    \begin{equation}
        \mathrm{Prob}_{\mathcal{T}_{n'}(D_m')} ( |\hat{f}'(x)-f(x)|< c\,\eta(m))\geq \mathrm{Prob}_{\mathcal{T}_{n}(D_m)} ( |\hat{f}(x)-f(x)|< c\,\eta(m))- \phi
    \end{equation}
    for any $\eta(m)$ and $c\in \mathbb{R}^+$.
    \end{lemma}
    
    \begin{proof}
    Note that 
    \begin{align}
        \mathrm{Prob} \left( \epsilon'(x) < c\,\eta(m)\right)&\geq \mathrm{Prob} \left( \epsilon'(x) < c\,\eta(m)| \epsilon'(x)\leq  \epsilon(x)\right)\mathrm{Prob}(\epsilon'(x)\leq \epsilon(x))\nonumber\\&
        \geq \mathrm{Prob} \left( \epsilon(x) < c\,\eta(m)| \epsilon'(x)\leq  \epsilon(x)\right)\mathrm{Prob}(\epsilon'(x)\leq \epsilon(x))
    \end{align}
    where $\epsilon'(x):=|\hat{f}'(x)-f(x)|$ and $\epsilon(x):=|\hat{f}(x)-f(x)|$. The first step follows from the law of total probability, and the second from the inequality $ \mathrm{Prob} \left( \epsilon'(x) < c\,\eta(m)| \epsilon'(x)\leq  \epsilon(x)\right)\geq \mathrm{Prob} \left( \epsilon(x) < c\,\eta(m)| \epsilon'(x)\leq  \epsilon(x)\right)$.
    The previous expression can be rewritten as 
    \begin{align}
        \mathrm{Prob} &\left( \epsilon(x) < c\,\eta(m)\right)\nonumber\\&-\mathrm{Prob} \left( \epsilon(x) < c\,\eta(m)| \epsilon'(x)>  \epsilon(x)\right)\mathrm{Prob}(\epsilon'(x)> \epsilon(x))
    \end{align}
    Next, using that $\mathrm{Prob} \left( \epsilon(x) < c\,\eta(m)| \epsilon'(x)>  \epsilon(x)\right)\leq 1$ and $\mathrm{Prob}(\epsilon'(x)> \epsilon(x))\leq \phi$, the following lower bound is obtained
    \begin{equation}
         \mathrm{Prob} \left( \epsilon(x) < c\,\eta(m)\right)-\phi
    \end{equation}

    \end{proof}

    Now, we are ready to start with the proof of Proposition \ref{prop:suff_inf}.
    
    \begin{proof}
        For simplicity in the notation and without loss of generality, we assume that $|\{x\in \mathcal{X}: b(x)\leq m\}|=2^m$. Note that if all binary sequences represent different values in the domain of the function, then $|\{x\in \mathcal{X}: b(x)\leq m\}|=2^{m+1}-1$.
        
        If $f\in \mathcal{L}_C^{\eta(m)}$ ($\mathcal{L}_Q^{\eta(m)}$), then there exists a classical (quantum) learning procedure, a sequence of distributions $\mathcal{D}$, an a polynomial $p\left(m\right)$ such that taking $n=p\left(m\right)$,
        \begin{equation}
            \mathrm{Prob}_{\mathcal{T}_n(D_m)}(|\hat{f}(x)-f(x)|\leq c\, \eta(m)) \geq 1-\delta
        \end{equation}
        for some constants $c\in \mathbb{R}^+$ and $\delta\in [0,\frac{1}{2})$, and for any $x\in \{x\in \mathcal{X}:b(x)\leq m\}$. Furthermore, if the function satisfies the classical (quantum) invariance property, then there exists a cover $\{R_j\}_{j=1}^{r(m)}$ fulfilling the conditions introduced in Definition \ref{def:invariance_property}.
        Note that $N_{R_j}(\mathcal{T}_n(D_m))\leq p\left(m\right)$ for any realization of the training set. Therefore, if we sample from a sequence of distributions $\mathcal{D}'$ such that
        \begin{equation}\label{eq:prop_P}
            \mathrm{Prob}_{\mathcal{T}_{n'}(D'_m)}\left(\bigcap_{j=1}^{r(m)} N_{R_j}(\mathcal{T}_{n'} (D'_m) )\geq p\left(m\right)\right)\geq 1-\phi
        \end{equation}
        where $n'=h(m,\frac{1}{\phi})$, and $h(a,b)$ is a multivariate polynomial, then 
        \begin{equation}
            \mathrm{Prob}_{\mathcal{T}_n(D_m),\mathcal{T}_{n'}(D_m')} \left(\bigcap_{j=1}^{r\left(m \right)} N_{R_j}(\mathcal{T}_n(D_m))\leq N_{R_j}(\mathcal{T}_{n'}(D_m'))\right)\geq 1-\phi
        \end{equation}
        Using Lemma \ref{lemma:increase_prob} and the classical (quantum) invariance property, 
        \begin{equation}
            \mathrm{Prob}_{\mathcal{T}_{n'}(D_m')}\left(|\hat{f}'(x)-f(x)|\leq c\, \eta(m)\right) \geq 1-\delta-\phi
        \end{equation}
        Taking for example $\phi:=\frac{1}{2}\left(\frac{1}{2}-\delta\right)$, then
        \begin{equation}
            \mathrm{Prob}_{\mathcal{T}_{n'}(D_m')}(|\hat{f}'(x)-f(x)|\leq c\, \eta(m)) \geq 1-\delta'
        \end{equation}
        where $\delta'= \frac{\delta}{2}+\frac{1}{4}<\frac{1}{2}$, and $n'=\tilde{h}(m)=h(m,\left(\frac{1}{2}\left(\frac{1}{2}-\delta\right)\right)^{-1})$. Therefore, $f(x)$ is classical (quantum) $\eta$-learnable using $\mathcal{D}'$. 
        
        On the other hand, note that any distribution $D_m'$ that satisfies \eqref{eq:prop_suff_inf} also fulfills
        \begin{align}
            \mathrm{Prob}_{X_{n+1}\sim D_m',\mathcal{T}_n (D_m')}&\left(x_{n+1}\in R_j\text{ and }(x_{n+1},f(x_{n+1}))\notin \mathcal{T}_n (D_m') \right)\nonumber\\& \geq \frac{1}{2^m q(m)} \left(|R_j|-N_{R_j}(\mathcal{T}_n (D_m'))\right)\geq \frac{\frac{2^m}{r\left(m\right)}-p\left(m\right)}{2^m q(m)} \nonumber\\& = \frac{1}{q(m)r\left(m\right)}-\frac{p(m)}{2^m q(m)}:= \frac{1}{z(m)}
        \end{align}
        where the first inequality follows from \eqref{eq:prop_suff_inf}. Hence, the random walk $W_{i+1}=W_i+V$, where $V\sim \mathrm{Bern}\left(1/z(m)\right)$, grows slower than
        \begin{equation}
            N_{R_j}(\mathcal{T}_{i+1}(D_m))=N_{R_j}(\mathcal{T}_{i}(D_m))+\mathbbm{1}\left\{ x_{i+1}\in R_j \text{ and }(x_{i+1},f(x_{i+1}))\notin \mathcal{T}_{i}(D_m)\right\}
        \end{equation}
        This means that the expected number of steps that $W_i$ needs to reach $p(m)$ is larger than the expected number of steps for $N_{R_j}(\mathcal{T}_{i}(D_m))$. We denote the latter as $\mathbb{E}[\tilde{n}]$, which satisfies
        \begin{equation}
            \mathbb{E}[\tilde{n}]\leq  \frac{p(m)-1}{1/z\left(m\right)}= z\left(m\right) \left(p\left(m\right)-1\right)
        \end{equation}
        where the right hand side is the expected number of steps for the random walk $W_i$ to reach $p(m)$. Therefore, using Markov's inequality
        \begin{equation}
                    \frac{1}{K}\geq \mathrm{Prob} \left( \tilde{n}\geq K \mathbb{E}[\tilde{n}]\right)\geq \mathrm{Prob} \left( \tilde{n}\geq K z\left(m\right) \left(p\left(m\right)-1\right)+1\right)
        \end{equation}
        This implies that if we take $n'=K z\left(m\right) \left(p\left(m\right)-1\right)$ samples from the distribution $D'_m$, then with probability at least $1-\frac{1}{K}$ the number of distinct elements in $R_j$ is higher or equal to $p(m)$, or equivalently,  
        \begin{equation}
            1-\frac{1}{K}\leq \mathrm{Prob} \left( \tilde{n}< n'\right) \implies \mathrm{Prob}_{\mathcal{T}_{n'}(D_m')} \left( N_{R_j}(\mathcal{T}_{n'}(D_m'))< p\left(m\right)\right)\leq \frac{1}{K} 
        \end{equation}
        Using the previous result, the probability that some inequality does not hold can be upper bounded as  
        \begin{align}
            \mathrm{Prob} \left(\bigcup_{j=1}^{r\left(m\right)} N_{R_j}(\mathcal{T}_{n'}(D_m'))<p\left(m\right)\right)&\leq \sum_{j=1}^{r\left(m\right)} \mathrm{Prob} \left( N_{R_j}(\mathcal{T}_{n'}(D_m')))<p\left(m\right)\right)\nonumber \\ & \leq \frac{r\left(m\right)}{K}
        \end{align}
        Taking $K=K' r\left(m\right)$, we have that 
        \begin{equation}
            \mathrm{Prob} \left(\bigcup_{j=1}^{r\left(m\right)} N_{R_j}(\mathcal{T}_{n'}(D_m'))<p\left(m\right)\right)\leq \frac{1}{K'}
        \end{equation}
        Note that $\frac{1}{K'}=\phi$, which implies that taking $n'=\frac{2}{\frac{1}{2}-\delta}r\left(m\right)z(m)(p\left(m\right)-1)$, then
        \begin{equation}
            \mathrm{Prob}_{\mathcal{T}_{n'}(D_m')}(|\hat{f}(x)-f(x)|\leq c \eta(m)) \geq 1-\delta'
        \end{equation}
        where $\delta'\in[0,\frac{1}{2})$. Finally, for a sufficiently large $m$, the number of elements of the training set can be upper bounded as
        \begin{equation}
            n'=\frac{2r(m)^2}{\frac{1}{2}-\delta} \frac{p(m)q(m)}{1-\frac{p\left(m\right)}{|R_j|}}\leq \frac{4\, r\left(m\right)^2}{\frac{1}{2}-\delta} p\left(m \right)q(m)
        \end{equation}
        which is a polynomial in $m$. The inequality follows since $p(m)<\left(\frac{1}{2}\right)\frac{2^m}{r\left(m\right)}$ for a sufficiently large $m$.

    \end{proof}

 \section{Proof of Theorem 3}\label{proof_theorem_3}

    This section shows the proof of Theorem \ref{thm:non_distribution}. That is, we prove the inclusions \eqref{classical_invariance_result} and \eqref{quantum_invariance_result}. For the expression \eqref{classical_invariance_result}, we start by using Theorem \ref{thm:conditions_speed_up}, 
    \begin{align}
        \left(\mathcal{P}_Q^{\,\eta}\backslash \mathcal{P}_C^{\,\eta}\right)\cap \mathcal{G}_C(\mathcal{D})&=\left(\mathcal{L}_Q^{\,\eta}(\mathcal{D})\backslash \mathcal{L}_C^{\,\eta}(\mathcal{D})\right)\cap \mathcal{G}_C(\mathcal{D})\nonumber \\& \subseteq \left(\mathcal{L}_Q^{\,\eta}\backslash \mathcal{L}_C^{\,\eta}(\mathcal{D})\right)\cap \mathcal{G}_C(\mathcal{D})
    \end{align}
    where the inclusion follows from $\mathcal{L}_Q^{\,\eta}(\mathcal{D})\subseteq \mathcal{L}_Q^{\,\eta}$. Therefore,
    \begin{equation}
        \left(\mathcal{P}_Q^{\,\eta}\backslash \mathcal{P}_C^{\,\eta}\right)\cap \mathcal{G}_C(\mathcal{D})\cap \mathcal{I}_C\subseteq \left(\mathcal{L}_Q^{\,\eta}\backslash \mathcal{L}_C^{\,\eta}(\mathcal{D})\right)\cap \mathcal{G}_C(\mathcal{D})\cap \mathcal{I}_C
    \end{equation}
    where $\mathcal{I}_C$ denotes the set of classical invariant functions.
    Next, by considering a sequence $\mathcal{D}$ that satisfies \eqref{eq:prop_suff_inf}, and using Proposition \ref{prop:suff_inf},
    \begin{align}
        \left(\mathcal{P}_Q^{\,\eta}\backslash \mathcal{P}_C^{\,\eta}\right)\cap \mathcal{G}_C(\mathcal{D})\cap \mathcal{I}_C &\subseteq \left(\mathcal{L}_Q^{\,\eta}\backslash \mathcal{L}_C^{\,\eta}(\mathcal{D})\right)\cap \mathcal{G}_C(\mathcal{D})\cap \mathcal{I}_C \nonumber \\ & \subseteq \left(\mathcal{L}_Q^{\,\eta}\backslash \mathcal{L}_C^{\,\eta}\right)\cap \mathcal{G}_C(\mathcal{D})\cap \mathcal{I}_C 
    \end{align}
    Therefore, by taking the union over all sequences that satisfy \eqref{eq:prop_suff_inf}, we obtain 
    \begin{equation}
        \bigcup_{\mathcal{D}}\bigg(\left(\mathcal{P}_Q^{\,\eta}\backslash \mathcal{P}_C^{\,\eta}\right)\cap \mathcal{G}_C(\mathcal{D})\cap \mathcal{I}_C \bigg )\subseteq \bigcup_{\mathcal{D}}\bigg(\left(\mathcal{L}_Q^{\,\eta}\backslash \mathcal{L}_C^{\,\eta}\right)\cap \mathcal{G}_C(\mathcal{D})\cap \mathcal{I}_C \bigg)
    \end{equation}
   Next, using that $(A\cap B)\cup (A \cap C)= A\cap (B\cup C)$,
    \begin{equation}
        \left(\mathcal{P}_Q^{\,\eta}\backslash \mathcal{P}_C^{\,\eta}\right)\cap \mathcal{G}_C\cap \mathcal{I}_C\subseteq \left(\mathcal{L}_Q^{\,\eta}\backslash \mathcal{L}_C^{\,\eta}\right)\cap \mathcal{G}_C\cap \mathcal{I}_C
    \end{equation}
    where $\mathcal{G}_C=\bigcup_{\mathcal{D}} \mathcal{G}_C(\mathcal{D}) $, i.e., the set of functions that are classical generable for some sequence $\mathcal{D}$ that satisfies \eqref{eq:prop_suff_inf}.

    Now, for the inclusion \eqref{quantum_invariance_result}, note that if $f \in \left(\mathcal{L}_Q^{\,\eta}\backslash \mathcal{L}_C^{\,\eta}\right)\cap \mathcal{G}_C\cap \mathcal{I}_Q$, then from Proposition \ref{prop:suff_inf} it follows that there exists a sequence $\mathcal{D}$ that satisfies \eqref{eq:prop_suff_inf}, such that $f\in \left(\mathcal{L}_Q^{\,\eta}(\mathcal{D})\backslash \mathcal{L}_C^{\,\eta}\right)\cap \mathcal{G}_C(\mathcal{D})\cap \mathcal{I}_Q$. Therefore, $f\in \left(\mathcal{L}_Q^{\,\eta}(\mathcal{D})\backslash \mathcal{L}_C^{\,\eta}(\mathcal{D})\right)\cap \mathcal{G}_C(\mathcal{D})\cap \mathcal{I}_Q$ for some sequence $\mathcal{D}$ that satisfies \eqref{eq:prop_suff_inf}. Next, using Theorem \ref{thm:conditions_speed_up}, $f\in \left(\mathcal{P}_Q^{\,\eta}\backslash \mathcal{P}_C^{\,\eta}\right)\cap \mathcal{G}_C(\mathcal{D})\cap \mathcal{I}_Q$. Finally, using the definition of $\mathcal{G}_C$, $f\in \left(\mathcal{P}_Q^{\,\eta}\backslash \mathcal{P}_C^{\,\eta}\right)\cap \mathcal{G}_C\cap \mathcal{I}_Q$. Consequently, 
    \begin{equation}
         \left(\mathcal{L}_Q^{\,\eta}\backslash \mathcal{L}_C^{\,\eta}\right)\cap \mathcal{G}_C\cap \mathcal{I}_Q \subseteq \left(\mathcal{P}_Q^{\,\eta}\backslash \mathcal{P}_C^{\,\eta}\right)\cap \mathcal{G}_C\cap \mathcal{I}_Q
    \end{equation}

\newpage

    \section{Proof of Theorem \ref{thm:examples}}\label{sec:proof_details_prime}

    \setcounter{equation}{0}
    
    This section shows some details of the proof of Theorem \ref{thm:examples}. The section is divided into two parts. First, the reduction between computing approximately the examples and factoring integers of the form $x=p_1^{r_1}\cdot p_2^{r_2}$ is presented. Next, we prove that there exists an efficient classical algorithm that samples from a distribution that satisfies the condition introduced in Proposition \ref{prop:suff_inf}.

    \subsection{Reduction to the prime factorization problem}\label{Reduction algorithms}

    The reduction argument consists in showing that if function $f_1(x)= \sum_{i=1}^{\omega(x)} p_i$, or both 
    \begin{equation}\label{eq:function_speed_up2}
    \begin{matrix}f_2(x)= \prod_{i=1}^{\omega(x)} p_i \text{\,\, and} & f_3(x)=\left[\sqrt{p_1 \sqrt{p_2\cdots \sqrt{p_{\omega(x)}}}}\right]^4 
     \\
    \end{matrix}
    \end{equation}
defined over the domain $\mathcal{X}=\{x\in \mathbb{N}:\omega(x)\leq K\}$, belong to $\mathcal{P}_C^{\,\eta}$ for $\eta(m)=m^u$, then we can factorize any integer of the form $x=p_1^{r_1}\cdot p_2^{r_2}$ in polynomial time with a classical algorithm.

First, if the previous functions belong to $\mathcal{P}_C^{\,\eta}$, then there exists an efficient algorithm that estimates $f_l(x)$ with an error $O(\log(x)^u)$ with at least probability $1/2$. Furthermore, as the functions $f_1(x)$ and $f_2(x)$ satisfy that $f_l(x)\in \mathbb{N}$, and $f_3(x)\in \mathbb{N}$ if $x=p_1^{r_1}\cdot p_2^{r_2}$, then $f_l(x)\in \mathbb{N}\cap [\hat{f}_l(x)-c \log(x)^u,\hat{f}_l(x)+c \log(x)^u]$ for some constant $c$ with at least probability $1/2$. Note that the number of integers in that interval is given by $\left \lceil 2 c \log(x)^u\right \rceil$, i.e., there are only a polynomial number of potential candidates. Therefore, if we can check the existence of a solution to the system of equations,
    \begin{equation}
        \left\{\begin{matrix}
        \tilde{f}=f_i(p_1 \cdot p_2) \\
        x=p_1^{r_1}\cdot p_2^{r_2}
        \end{matrix}\right.
    \end{equation}
where $p_1,p_2\in\mathbb{P}$, and $r_1,r_2\in \mathbb{N}$, then $f_l(x)$ can be computed in polynomial time with at least probability 1/2. Furthermore, in the case that $\tilde{f}=f_l(x)$, during the verification process the prime factors $p_1$ and $p_2$ are obtained. The probability can be done arbitrarily close to 1 by reaping this procedure a constant number of times. Note also that for each candidate $\tilde{f}$ we need to check at most $\log(x)^2$ different values of $r_1$ and $r_2$, since $r_1,r_2 \leq \log(x)$.  

To check if the system of equations has a solution for function $f_1(x)$, first, we compute $\bar{r}=\mathrm{gcd}(r_1,r_2)$, which is used to simplify the algorithm. Next, combining $f_1(x)=p_1+p_2$ and $x^{1/\bar{r}}=p_1^{r_1/\bar{r}} \cdot p_2^{r_2/\bar{r}}$, we have that $p_1$ is a root of
\begin{equation}
    p(u)=u^{\frac{r_1}{\bar{r}}} (f_1-u)^{\frac{r_2}{\bar{r}}}-x^{\frac{1}{\bar{r}}}.
\end{equation}
As the polynomial $q(u)=u^{\frac{r_1}{\bar{r}}} (f_1-u)^{\frac{r_2}{\bar{r}}}$ has a parabolic shape, there are at most two real roots of $p(u)$ in the interval $(0,f_1)$. To find them, we use a binary search method in the intervals $I_-=(0,  u_{\max} )$ and $I_+=(u_{\max} , f_1)$, where $u_{\max}=r_1 f_1/(r_1+r_2)$. Note that in each iteration the length of the intervals is divided by two. Therefore, $\log f_1\leq \log 2x$ iterations are sufficient to reduce the length of the intervals to less than 1. At this point, we can compute both integers that belong to the resulting intervals and check if one of them is a prime factor of $x$. Therefore, if a prime factor of $x$ is obtained, then $\tilde{f}=f_1(x)$, otherwise $\tilde{f}\neq f_1(x)$. Consequently, if $f_1\in \mathcal{P}_C^{\,\eta}$, we can factorize numbers of the form $p_1^{r_1}\cdot p_2^{r_2}$ in polynomial time with high probability. The pseudo-code of the resulting algorithm is shown in Algorithm \ref{alg_f1}.

To verify if $\tilde{f}=f_2(x)$, the fact that $f_2(x)|x$ is used. Moreover, after dividing $x$ by $f_2(x)$ as many times as possible, we obtain $z=p^{|r_1-r_2|}$, where $p\in \{p_1,p_2\}$. Note that numbers with only one prime factor can be factorized efficiently. Hence, if $\tilde{f}=f_2(x)$, then $z=1$ or a prime factor of $x$ is obtained. Therefore, if $f_2\in \mathcal{P}_C^{\,\eta}$, then we can factorize any number of the form $p_1^{r_1}\cdot p_2^{r_2}$ such that $r_1\neq r_2$.

Next, to check if $\tilde{f}=f_3(x)$, we use that if $r_1<2r_2$, then $x^2/f_3^{r_1}=p_2^{2r_2-r_1}\in \mathbb{N}\backslash \{0,1\}$, which can be factorized efficiently. Similarly, if $r_1>2r_2$ then $x/f_3^{r_2}=p_1^{r_1-2r_2}\in \mathbb{N}\backslash \{0,1\}$. Therefore, if $f_3(x)\in \mathcal{P}_C^{\,\eta}$, then there exists an algorithm to factorize numbers of the form $p_1^{r_1} \cdot p_2^{r_2}$ such that $r_1\neq 2r_2$ in polynomial time.    

Finally, if both $f_2(x)$ and $f_3(x)$ belong to $\mathcal{P}_C^{\,\eta}$, then we can factorize any number of the form $x=p_1^{r_1} \cdot p_2^{r_2}$ in polynomial time with high probability. The pseudo-code is shown in Algorithm \ref{alg2_f2_f3}.

    \begin{algorithm}[H]
     \caption{Factorization using an estimate of function $f_1(x)$}
    \label{alg_f1}
    \begin{algorithmic}
    	\State\textbf{Input:} $x=p_1^{r_1}\cdot p_2^{r_2}$
    	\State\textbf{Output:} $p_1$
    	\State Compute estimate $\hat{f}_1(x)$
    	\State cont=0, check=False
    	\State $\tilde{f}=f_1(x)$

    	\While{cont$<\left \lceil 2c (\log x)^u \right \rceil$ and check=False} 
    	    \For{$r_1'=1:\log x$}
    	        \For{$r_2'=1:\log x$}
    	        
    	            \State $u_{\max}=\frac{r_1' \tilde{f}}{r_1'+r_2'}$
    	            
    	            \State Using binary search, compute the nearest integer to the roots of $P(u)$ in 
    	            \State the intervals $[0,u_{\max}]$ and $[u_{\max}, \tilde{f}]$.
    	            
    	            \State Assign the values to variables $h_-$ and $h_+$, respectively.
            	    \If{$h_-\in \mathbb{P}$ and $h_-|x$}
            	        \State check=True
            	        \State $p_1=h_-$
            	   \EndIf
            	   
                   \If{$h_+\in \mathbb{P}$ and $h_+|x$}
            	        \State check=True
            	        \State $p_1=h_+$
            	   \EndIf
            	   
            	\EndFor
        	\EndFor
        	\State cont=cont+1
        	\State $\tilde{f}=\tilde{f}+(-1)^{\mathrm{cont}}\mathrm{cont}$
        	
    	\EndWhile
	\end{algorithmic}
\end{algorithm}

    \begin{algorithm}[H]
     \caption{Factorization using an estimate of functions $f_2(x)$ and $f_3(x)$}
    \label{alg2_f2_f3}
    \begin{algorithmic}
    	\State\textbf{Input:} $x=p_1^{r_1}\cdot p_2^{r_2}$
    	\State\textbf{Output:} $p\in\{p_1, p_2\}$
    	\State Compute estimate $\hat{f}_2(x)$
    	\State cont=0, check=False
    	\State $\tilde{f}=f_2(x)$

    	\While{cont$<\left \lceil 2c (\log x)^u \right \rceil$ and check=False} 
    	    \State $z=x$
    	    
    	    \While{ $\tilde{f}$ divides $z$}
    	        \State $z=z/\tilde{f}$
     	    \EndWhile 
     	    
     	    \If{$z=1$}
     	        \State check=True
     	    \Else
     	        \If{$z=\hat{p}^r$ such that $\hat{p}\in \mathbb{P}$, and $\hat{p}|x$}
     	            \State check=True
     	            \State $p=\hat{p}$
     	        \EndIf
     	    \EndIf
    	    
        	\State cont=cont+1
        	\State $\tilde{f}=\tilde{f}+(-1)^{\mathrm{cont}}\mathrm{cont}$
        	
    	\EndWhile
    	
    	\If{$z=1$}
    	    \State Compute estimate $\hat{f}_3(x)$
    	    \State cont=0, check=False
    	    \State $\tilde{f}=f_3(x)$
    	    
    	    \While{cont$<\left \lceil 2c (\log x)^u \right \rceil$ and check=False} 
                \For{$r_1'=1:\log(x)$}
                    \For{$r_2'=1:\log(x)$}
                        \If{$x^2/\tilde{f}^{r_1'}=\hat{p}^r$ such that $\hat{p}\in\mathbb{P}$ and $\hat{p}|x$}
                            \State $p=\hat{p}$
                            \State check=True
                        \EndIf
                    \EndFor
                \EndFor
        	    
            	\State cont=cont+1
            	\State $\tilde{f}=\tilde{f}+(-1)^{\mathrm{cont}}\mathrm{cont}$
            	
        	\EndWhile
    	    
    	\EndIf
	\end{algorithmic}
\end{algorithm}

\subsection{Sampling Algorithm} \label{sampling_algorithm_section}

 Next, we show that there exists an efficient classical algorithm to sample from a sequence of distributions $\mathcal{D}$ that satisfies \eqref{eq:prop_suff_inf}. As mentioned in Section \ref{sec:realtion_learning_and_computational}, the algorithm is divided in three steps. 
 
 \subsubsection{Step 1: Number of prime factors}
 
 First, the number of prime factors, denoted by $\omega$, is randomly selected. Note that if we sample from the uniform distribution over $\{x:\omega(x)\leq K\text{ and }b(x)\leq m\}$, then the probability that $x$ has $\omega$ different prime factors satisfies
\begin{equation}
    \text{Prob}(\omega)\propto \pi_{\omega}(2^m)
\end{equation}
where $\omega\in[1:K]$, and $\pi_\omega(2^m)$ denotes the number of integers smaller than $2^m$ with exactly $\omega$ different prime factors. Since $\pi_{\omega}(2^m)$ cannot be efficiently computed, we use  the following result from Landau \cite{landau1909handbuch}

\begin{equation}
    \pi_{\omega}(u) \sim  \frac{u}{\log u} \frac{(\log\log u)^{\omega-1}}{(\omega-1)!}
\end{equation}

Hence, there exists a constant $m_0$ such that for all $m\geq m_0$

\begin{equation}\label{inequality_1}
    \frac{9}{11} \frac{\pi_\omega(2^m)}{\sum_{i=1}^K \pi_\omega(2^m)}\leq \frac{\frac{(\log m)^{\omega-1}}{(\omega-1)!}}{\sum_{i=1}^K \frac{(\log m)^{\omega-1}}{(\omega-1)!}}\leq \frac{11}{9} \frac{\pi_\omega(2^m)}{\sum_{i=1}^K \pi_\omega(2^m)}
\end{equation}

where the value $\frac{9}{11}$ is arbitrary, i.e., it can be made arbitrarily close to $1$ by increasing the constant $m_0$.

The previous equation implies that this approximation only deviates in a constant factor from the desired distribution. Furthermore, we can sample efficiently from this distribution. To do this, the expression $\frac{(\log m)^{\omega-1}}{(\omega-1)!}$ is computed for $\omega\in[1:K]$. Once the $K$ values are given, we can sample from the discrete distribution in time $O(K)$\footnote{Factors with a sub-polynomial growth in $m$ are omitted.} \cite{bringmann2017efficient}.

\subsubsection{Step 2: Prime factors}\label{sampling_algorithm_step_2}

After the number of different primes is determined, the set of $\omega$ prime factors is randomly selected. Similarly to the previous section, if we sample from the uniform distribution over $\{x:\omega(x)=\omega \text{ and }b(x)\leq m\}$, then the probability of observing  $p_1<p_2<\cdots<p_\omega$ is proportional to the number of integers which those $\omega$ prime factors, i.e.,
\begin{equation}\label{distribution_ideal_step_2}
    \mathrm{Prob}(p_1,\cdots,p_\omega)=\frac{|N_{p_1,\cdots,p_\omega}(2^m)|}{\pi_{\omega}(2^m)}
\end{equation}
where $N_{p_1,\cdots,p_\omega}(2^m)=\{x\in \mathbb{N}:\omega(x)=\omega,\, x\leq 2^m,\text{ and } x=\prod_{i=1}^\omega p_i^{r_i}\}$. In order to sample efficiently from a similar distribution, we need to make several approximations. In particular, we start by using the following lemma.

\begin{lemma}\label{lem:number_volume}
    The number of integers smaller than $2^m$, whose prime factors are $p_1,\cdots,p_\omega$, i.e., $|N_{p_1,\cdots,p_\omega}(2^m)|$, satisfies,
    \begin{equation}
       | N_{p_1,\cdots,p_\omega}(2^m)|\sim \frac{m^\omega}{\omega!}\prod_{i=1}^{\omega}\frac{1}{\log p_i}
    \end{equation}
\end{lemma}
\begin{proof}
    Note that any $x \in N_{p_1,\cdots,p_\omega}(2^m)$ fulfills
\begin{equation}
    1< \prod_{i=1}^\omega p_i \leq x = \prod_{i=1}^\omega p_i^{r_i}\leq 2^m
\end{equation}
Hence, $0<r_1\log p_1+\cdots+r_\omega\log p_\omega\leq m$, or equivalently, $r=[r_1,r_2,\cdots, r_\omega]^T$ belongs to a simplex, which we denote by $S_{p_1,\cdots,p_\omega}$, and whose vertices are 
\begin{equation}
    v_0=[0,0,\cdots,0]^T, v_1=\left[\frac{m}{\log p_1},0,\cdots,0\right]^T,\cdots, v_\omega=\left[0,\cdots, 0,\frac{m}{\log p_\omega}\right]^T
\end{equation}
$S_{p_1,\cdots,p_\omega}$ is related with $N_{p_1,\cdots,p_\omega}(2^m)$ as shown below, 
\begin{equation}
    |S_{p_1,\cdots,p_\omega}\cap\mathbb{Z}^\omega|\geq |N_{p_1,\cdots,p_\omega}(2^m)|\geq |S_{p_1,\cdots,p_\omega}\cap\mathbb{Z}^\omega|-|\delta S_{p_1,\cdots,p_\omega}\cap\mathbb{Z}^\omega|
\end{equation}
where $\delta S_{p_1,\cdots,p_\omega}$ denotes the surface of the simplex. Moreover, $|S_{p_1,\cdots,p_\omega}\cap \,\mathbb{Z}^\omega|\sim\text{Vol}(S_{p_1,\cdots,p_\omega})$, and $|\delta S_{p_1,\cdots,p_\omega}\cap \,\mathbb{Z}^\omega|\sim\text{Vol}(\delta S_{p_1,\cdots,p_\omega})$, and since the volume of the surface grows slower than the volume of the simplex, 
\begin{equation}
    |N_{p_1,\cdots,p_\omega}(2^m)|\sim\text{Vol}(S_{p_1,\cdots,p_\omega})=\frac{m^\omega}{\omega!}\prod_{i=1}^{\omega}\frac{1}{\log p_i}
\end{equation}
\end{proof}
Using Lemma \ref{lem:number_volume}, for a sufficiently large $m$,
\begin{equation}\label{inequality_2}
    \frac{9}{11} \frac{|N_{p_1,\cdots,p_{\omega}}(2^m)|}{\pi_{\omega}(2^m)}\leq \frac{\prod_{i=1}^{\omega}\frac{1}{\log p_i}}{\sum_{p_1<\cdots<p_{\omega}}\prod_{i=1}^{\omega}\frac{1}{\log p_i}} \leq \frac{11}{9} \frac{|N_{p_1,\cdots,p_{\omega}}(2^m)|}{\pi_{\omega}(2^m)}
\end{equation}
which implies that this approximation only deviates in a constant factor from distribution \eqref{distribution_ideal_step_2}\footnote{ Interestingly, since $\mathbbm{1}\{p_1\cdots p_\omega\leq 2^m\}\leq N_{p_1,\cdots,p_\omega}(2^m)\leq \mathbbm{1}\{p_1\cdots p_\omega\leq 2^m\}\cdot m^K$,
\begin{equation}\label{inequality_3_v2}
    \frac{1}{m^K} \frac{|N_{p_1,\cdots,p_{\omega}}(2^m)|}{\pi_{\omega}(2^m)}\leq \frac{\mathbbm{1}\{p_1\cdots p_\omega\leq 2^m\}}{\sum_{p_1<\cdots<p_{\omega}}\mathbbm{1}\{p_1\cdots p_\omega\leq 2^m\}} \leq m^K \frac{|N_{p_1,\cdots,p_{\omega}}(2^m)|}{\pi_{\omega}(2^m)}
\end{equation}
Therefore, this distribution deviates in a polynomial factor from $|N_{p_1,\cdots,p_{\omega}}(2^m)|/\pi_{\omega}(2^m)$.}. Next, we use Lemma \ref{lemma_ceiling}.
\begin{lemma}\label{lemma_ceiling}
    The following identity holds
    \begin{equation}
      \frac{1}{2^K} \prod_{i=1}^{\omega}\frac{1}{\log p_i} \leq  \prod_{i=1}^{\omega}\frac{1}{\left \lceil \log p_i \right \rceil}\leq \prod_{i=1}^{\omega}\frac{1}{\log p_i}
    \end{equation}
\end{lemma}
\begin{proof}

    Note that
    \begin{equation}
    \log p \leq \left \lceil \log p \right \rceil \leq \log p +1
    \end{equation}
    is satisfied for any $p\in \mathbb{P}$, where $\mathbb{P}$ denotes the set of prime numbers. Therefore, 
        \begin{equation}\label{inequality_ceil}
             1- \frac{1}{\left \lceil \log p \right \rceil} \leq \frac{\log p }{\left \lceil \log p \right \rceil} \leq  1 
        \end{equation}
    On the other hand, as $\log 2/\left \lceil \log 2 \right \rceil=1$, 
    \begin{equation}
        \min_{p\in \mathbb{P}} \left\{\frac{\log p}{ \left \lceil \log p \right \rceil} \right\} = \min_{p\in \mathbb{P}\backslash \{2\}} \left\{\frac{\log p}{ \left \lceil \log p \right \rceil} \right\}
    \end{equation}
    Using inequality \eqref{inequality_ceil},
        \begin{equation}
            \min_{p\in \mathbb{P}\backslash \{2\}} \left\{\frac{\log p}{ \left \lceil \log p \right \rceil} \right\}\geq \min_{p\in \mathbb{P}\backslash \{2\}} \left\{ 1- \frac{1}{\left \lceil \log p \right \rceil} \right\}= \frac{1}{2}
        \end{equation}
    Therefore,
    \begin{equation}
        1 \geq \prod_{i=1}^\omega \frac{\log p_i }{\left \lceil \log p_i \right \rceil} \geq \frac{1}{2^\omega} \geq \frac{1}{2^K}
    \end{equation}
\end{proof}
From this lemma follows that, 
\begin{equation}\label{inequality_3}
    \frac{1}{2^K}\frac{\prod_{i=1}^{\omega}\frac{1}{\log p_i}}{\sum_{p_1<\cdots<p_{\omega}}\prod_{i=1}^{\omega}\frac{1}{\log p_i}} \leq\frac{\prod_{i=1}^{\omega} \frac{1}{\left \lceil \log p_i \right \rceil}}{\sum_{p_1<\cdots<p_{\omega}}\prod_{i=1}^{\omega}\frac{1}{\left \lceil \log p_i \right \rceil}}  \leq 2^K \frac{\prod_{i=1}^{\omega}\frac{1}{\log p_i}}{\sum_{p_1<\cdots<p_{\omega}}\prod_{i=1}^{\omega}\frac{1}{\log p_i}}
\end{equation}
or equivalently, distribution 
\begin{equation}
    \mathrm{Prob}(p_1,\cdots,p_\omega)\propto \frac{1}{\left \lceil \log p_1 \right \rceil \cdots  \left \lceil \log p_\omega \right \rceil}
\end{equation}
where $p_i \in \mathbb{P}$, $p_i< p_j$ for all $i<j$, and $\sum_{i=1}^\omega \log p_i \leq m$, deviates in a constant factor from the desired distribution. Next, we modify the last condition to $\sum_{i=1}^\omega \lceil\log p_i \rceil \leq m+\omega$, and reject any sample that does not hold $\sum_{i=1}^\omega \log p_i \leq m$. As shown later, the rejection probability is bounded by a constant. On the other hand, now we can use $l_i=\lceil\log p_i \rceil$ as the random variable. Specifically,
\begin{align}
    \mathrm{Prob}(l_1,\cdots,l_\omega)&\propto \prod_{j: v_l(j)\neq 0}\frac{1}{j^{v_l(j)}}\left(\frac{1}{v_l(j)!}\prod_{i=0}^{v_l(j)-1} \left(\pi(2^{j}-1)-\pi(2^{j-1}-1)-i\right) \right)
    \nonumber\\&=\prod_{j: v_l(j)\neq 0} \frac{1}{j^{v_l(j)}} \binom{\pi(2^{j}-1)-\pi(2^{j-1}-1)}{v_l(j)} 
\end{align}
where $v_l(j)=\sum_{i=1}^{\omega} \mathbbm{1}\{l_i=j\}$, and $\pi(u)$ denotes the prime-counting function. The previous expression only holds if $v_l(j)\leq \pi(2^{j}-1)-\pi(2^{j-1}-1)$ $\forall\, j\in[1:m]$, otherwise the probability is 0.

Note that the number of primes with $j$ bits, i.e.,
\begin{equation}
    \pi(2^{j}-1)-\pi(2^{j-1}-1)=\pi(2^{j})-\pi(2^{j-1}) \text{ for } j>2,
\end{equation}
cannot be computed efficiently. For this reason, we use $\frac{u}{\log u}\leq \pi(u)\leq 1.26\frac{u}{\log u}$ for $u>17$ \cite{rosser1962approximate}. Interestingly, from the prime number theorem, which states $\pi(u)\sim \frac{u}{\log u}$, it follows that the value $1.26$ can be done arbitrarily close to $1$ by increasing the lower bound on $u$. Using the previous inequalities, we have that
\begin{equation}
    2^j \left( \frac{0.37j-1}{j(j-1)}\right)\leq \pi(2^{j})-\pi(2^{j-1})\leq 2^j \left(\frac{0.76j-1.26}{j (j-1)} \right)
\end{equation}
for $j\geq 6$.  Therefore, 
\begin{equation}\label{app:inequality_binom}
    \binom{\left \lfloor 2^j \left( \frac{0.37j-1}{j(j-1)}\right)\right \rfloor}{v_l(j)}\leq \binom{\pi(2^{j})-\pi(2^{j-1})}{v_l(j)} \leq \binom{\left \lceil 2^j \left(\frac{0.76j-1.26}{j (j-1)} \right)\right \rceil}{v_l(j)}
\end{equation}
where $j\geq C_K:=\max\{6,c_K\}$, and $c_K$ is the smallest integer that satisfies 
\begin{equation}
    K \leq \left \lfloor 2^{c_k} \left( \frac{0.37c_K-1}{c_K(c_K-1)}\right)\right \rfloor
\end{equation}
This guarantees that all the binomial coefficients are well defined for any value of $v_l(j)\in[0,\omega]$. Moreover, using that $\left(\frac{a}{b}\right)^b\leq \binom{a}{b}\leq \left(\frac{a e}{b}\right)^b$, 
\begin{equation}
    1\geq \frac{\binom{\pi(2^{j})-\pi(2^{j-1})}{v_l(j)}}{\binom{\left \lceil 2^j \left(\frac{0.76j-1.26}{j (j-1)} \right)\right \rceil}{v_l(j)}}\geq \left(\frac{\left \lfloor 2^j \left( \frac{0.37j-1}{j(j-1)}\right)\right \rfloor}{\left \lceil 2^j \left(\frac{0.76j-1.26}{j (j-1)} \right) \right \rceil e }  \right)^{v_l(j)} \geq \Xi(K)^K
\end{equation}
where $\Xi(K)=\min_{j\geq C_K} \left(\frac{\left \lfloor 2^j \left( \frac{0.37j-1}{j(j-1)}\right)\right \rfloor}{\left \lceil 2^j \left(\frac{0.76j-1.26}{j (j-1)} \right) \right \rceil e }  \right)$. 
This implies that
\begin{equation*}
     \Xi(K)^{K^2}\frac{\prod_{j: v_l(j)\neq 0} \frac{1}{j^{v_l(j)}} \binom{\pi(2^{j}-1)-\pi(2^{j-1}-1)}{v_l(j)} }{\sum_{l_1,\cdots,l_{\omega}}\prod_{j: v_l(j)\neq 0}\frac{1}{j^{v_l(j)}} \binom{\pi(2^{j}-1)-\pi(2^{j-1}-1)}{v_l(j)} }
\end{equation*}
\begin{equation*}
    \leq \frac{\prod_{j: v_l(j)\neq 0}\frac{1}{j^{v_l(j)}} \left(\mathbbm{1}\{j< C_K\}\binom{\pi(2^{j}-1)-\pi(2^{j-1}-1)}{v_l(j)}+\mathbbm{1}\{j\geq C_K\}\binom{\left \lceil 2^j \left(\frac{0.76j-1.26}{j (j-1)} \right)\right \rceil}{v_l(j)}\right) }{\sum_{l_1,\cdots,l_{\omega}}\prod_{j: v_l(j)\neq 0} \frac{1}{j^{v_l(j)}} \left(\mathbbm{1}\{j< C_K\}\binom{\pi(2^{j}-1)-\pi(2^{j-1}-1)}{v_l(j)}+\mathbbm{1}\{j\geq C_K\}\binom{\left \lceil 2^j \left(\frac{0.76j-1.26}{j (j-1)} \right)\right \rceil}{v_l(j)}\right) } 
\end{equation*}
\begin{equation}\label{distribution_equation}
    \leq \left(\frac{1}{\Xi(K)^{K^2}}\right) \frac{\prod_{j: v_l(j)\neq 0} \frac{1}{j^{v_l(j)}} \binom{\pi(2^{j}-1)-\pi(2^{j-1}-1)}{v_l(j)} }{\sum_{l_1,\cdots,l_{\omega}}\prod_{j: v_l(j)\neq 0}\frac{1}{j^{v_l(j)}} \binom{\pi(2^{j}-1)-\pi(2^{j-1}-1)}{v_l(j)} }
\end{equation}
This last distribution is the one that we use to generate the prime factors. In particular, we first generate a vector $l=[l_1,\cdots,l_\omega]$ with a probability proportional to 
\begin{align*}
    \prod_{j: v_l(j)\neq 0}\frac{1}{j^{v_l(j)}} \left(\mathbbm{1}\{j< C_K\}\binom{\pi(2^{j}-1)-\pi(2^{j-1}-1)}{v_l(j)}+\mathbbm{1}\{j\geq C_K\}\binom{\left \lceil 2^j \left(\frac{0.76j-1.26}{j (j-1)} \right)\right \rceil}{v_l(j)}\right)
\end{align*}
To sample from this distribution, note that $l_i\leq m$, i.e., there are at most $m^\omega$ values. We compute the unnormalized probability for each sequence, and as shown previously, samples can be generated from this discrete distribution in time $O(m^\omega)$. Afterwards, for each value of $l_i$ we generate a uniformly random prime factor with $l_i$ bits, which can be done in time $\mathrm{poly}(l_i)$ \cite{joye2000efficient}. 

Finally, we need to check if $\sum_{i=1}^\omega \log p_i >m$, in which case the sequence is rejected, and the process is repeated. Now, we show that the rejection probability is bounded by a constant, i.e., a constant number of repetitions is sufficient to generate a sample with high probability.  
\begin{align}
    \mathrm{Prob}\left( \sum_{i=1}^\omega \log p_i >m\right)&\leq \mathrm{Prob} \left( m<\sum_{i=1}^\omega l_i \leq m+\omega \right)\nonumber\\&=\sum_{\varphi=m+1}^{m+\omega}\sum_{l\in S_{\varphi}} \mathrm{Prob}(l) \nonumber\\&= \sum_{\varphi=m+1}^{m+\omega}\sum_{l\in S_{\varphi}} \left ( \frac{\prod_{j: v_l(j)\neq 0} \frac{G_l(j)}{j^{v_l(j)}} }{\sum_{\varphi=1}^{m+\omega}\sum_{l\in S_{\varphi}}\prod_{j: v_l(j)\neq 0} \frac{G_l(j)}{j^{v_l(j)}} }\right)\nonumber\\
    &=1 -  \left ( \frac{\sum_{\varphi=1}^{m}\sum_{l\in S_{\varphi}}\prod_{j: v_l(j)\neq 0} \frac{G_l(j)}{j^{v_l(j)}}}{\sum_{\varphi=1}^{m+\omega}\sum_{l\in S_{\varphi}}\prod_{j: v_l(j)\neq 0} \frac{G_l(j)}{j^{v_l(j)}}}\right)
\end{align}
The first inequality uses that $\sum_{i=1}^\omega l_i\leq m\implies \sum_{i=1}^\omega \log p_i \leq m$, i.e., if $\sum_{i=1}^\omega \log p_i > m$ then $\sum_{i=1}^\omega l_i> m$. The second and third steps follow from definitions 
\begin{equation}
   S_\varphi:= \left\{ l\in \mathbb{N}^\omega: \sum_{i=1}^\omega l_i = \varphi \text{, }l_1\leq l_2 \leq \cdots \leq l_\omega \text{ and }  v_l(j)\leq \pi(2^{j}-1)-\pi(2^{j-1}-1) \,\forall j\in[1:m]\right \}
\end{equation}

and, 
\begin{equation}
    G_l(j):= \left(\mathbbm{1}\{j< C_K\}\binom{\pi(2^{j}-1)-\pi(2^{j-1}-1)}{v_l(j)}+\mathbbm{1}\{j\geq C_K\}\binom{\left \lceil 2^j \left(\frac{0.76j-1.26}{j (j-1)} \right)\right \rceil}{v_l(j)}\right) 
\end{equation}
The upper bound can be rewritten as,
\begin{equation}
    1 -  \left ( \frac{\sum_{\varphi=1}^{m}3^{\varphi}\left(\sum_{l\in S_{\varphi}} 3^{-\varphi}\prod_{j: v_l(j)\neq 0}  \frac{G_l(j)}{j^{v_l(j)}}\right)}{\sum_{\varphi=1}^{m+\omega}3^{\varphi}\left(\sum_{l\in S_{\varphi}}3^{-\varphi}\prod_{j: v_l(j)\neq 0}  \frac{G_l(j)}{j^{v_l(j)}}\right)}\right)
\end{equation}
Furthermore, defining $A(\varphi):=\sum_{l\in S_{\varphi}} 3^{-\varphi}\prod_{j: v_l(j)\neq 0}  \frac{G_l(j)}{j^{v_l(j)}}$, the previous expression becomes
\begin{equation}
    1 -  \left ( \frac{\sum_{\varphi=1}^{m}3^{\varphi} A(\varphi)}{\sum_{\varphi=1}^{m+\omega}3^{\varphi}A(\varphi)}\right)
\end{equation}
Hence, the probability of rejection can be expressed as
\begin{align}
    \lim_{m\rightarrow \infty}\mathrm{Prob}\left( \sum_{i=1}^\omega \log p_i >m\right)&\leq 1 -  \lim_{m\rightarrow \infty}\left ( \frac{\sum_{\varphi=1}^{m}3^{\varphi} A(\varphi)}{\sum_{\varphi=1}^{m+\omega}3^{\varphi}A(\varphi)}\right)\nonumber\\&=  1-\lim_{m\rightarrow \infty}\frac{\sum_{\varphi=1}^{m}3^{\varphi} A(\varphi)}{\sum_{\varphi=1}^{m+1}3^{\varphi} A(\varphi)}\frac{\sum_{\varphi=1}^{m+1}3^{\varphi} A(\varphi)}{\sum_{\varphi=1}^{m+2}3^{\varphi} A(\varphi)}\cdots\frac{\sum_{\varphi=1}^{m+\omega-1}3^{\varphi} A(\varphi)}{\sum_{\varphi=1}^{m+\omega}3^{\varphi} A(\varphi)}\nonumber\\&=1 -  \left ( \lim_{m\rightarrow \infty}\frac{\sum_{\varphi=1}^{m}3^{\varphi} A(\varphi)}{\sum_{\varphi=1}^{m+1}3^{\varphi}A(\varphi)}\right)^{\omega}
\end{align}
where the last step uses that 
\begin{equation}
    \lim_{m\rightarrow \infty}\left ( \frac{\sum_{\varphi=1}^{m}3^{\varphi} A(\varphi)}{\sum_{\varphi=1}^{m+1}3^{\varphi}A(\varphi)}\right)= \lim_{m\rightarrow \infty}\left ( \frac{\sum_{\varphi=1}^{m+c}3^{\varphi} A(\varphi)}{\sum_{\varphi=1}^{m+c+1}3^{\varphi}A(\varphi)}\right)
\end{equation}
for any constant $c\in \mathbb{N}$.

As it is shown later, $A(\varphi)=O(\left(\frac{2}{3}\right)^{\varphi}\varphi^{\omega-1})$ and  $A(\varphi)=\Omega(\left(\frac{2}{3}\right)^{\varphi}\varphi^{-2\omega})$, this implies that there exists a value $\varphi_0$ such that the function $A(\varphi)$ is decreasing for all $\varphi\geq \varphi_0$. Therefore,
\begin{align}
    3 \sum_{\varphi=\varphi_0}^{m} 3^{\varphi} A(\varphi)&= \sum_{\varphi=\varphi_0+1}^{m+1} 3^{\varphi} A(\varphi-1) \nonumber\\ & \geq  \sum_{\varphi=1}^{m+1} 3^{\varphi} A(\varphi)-\sum_{\varphi=1}^{\varphi_0}3^{\varphi} A(\varphi)
\end{align}
or equivalently, 
\begin{align}
    \frac{ \sum_{\varphi=1}^{m} 3^{\varphi} A(\varphi)}{\sum_{\varphi=1}^{m+1} 3^{\varphi} A(\varphi)}\geq  \frac{1}{3}+\frac{\sum_{\varphi=1}^{\varphi_0-1}3^{\varphi} A(\varphi)-\sum_{\varphi=1}^{\varphi_0}3^{\varphi-1} A(\varphi)}{\sum_{\varphi=1}^{m+1} 3^{\varphi} A(\varphi)}
\end{align}
and since $A(\varphi)=\Omega(\left(\frac{2}{3}\right)^{\varphi}\varphi^{-2\omega})$ then ${\sum_{\varphi=1}^{m} 3^{\varphi} A(\varphi)}\rightarrow \infty$ as $m\rightarrow \infty$. Using this, the probability of rejection can be upper bounded as 
\begin{equation}
    \lim_{m\rightarrow \infty}\mathrm{Prob}\left( \sum_{i=1}^\omega \log p_i >m\right) \leq  1- \frac{1}{3^\omega}  
\end{equation}
That is, the probability of rejection is bounded by a constant w.r.t. $m$.
The only thing left is demonstrating the bounds used for $A(\varphi)$. First, we prove the upper bound of $A(\varphi)$.  
\begin{align}
    A(\varphi)&\leq \sum_{l\in S_\varphi}3^{-\varphi} \prod_{\substack{j=1\\v_l(j)\neq 0}}^{C_K-1} \frac{1}{j^{v_l(j)}}  \left(e\cdot \frac{\pi(2^{j}-1)-\pi(2^{j-1}-1)}{v_l(j)}\right)^{v_l(j)} \prod_{\substack{j=C_K\\v_l(j)\neq 0}}^m \frac{1}{j^{v_l(j)}} \left(\frac{e \,\left \lceil 2^j \left(\frac{0.76j-1.26}{j (j-1)} \right)\right \rceil}{v_l(j)}\right)^{v_l(j)}\nonumber\\ &\leq \sum_{l\in S_{\varphi}} \left(\frac{2}{3}\right)^\varphi e^\omega \,\Phi_{\text{max}} ^{\sum_{j=1}^{C_K-1}v_l(j)} \prod_{l_i\geq C_K} \frac{1}{2^{l_i} l_i} \left\lceil 2^{l_i} \left(\frac{0.76l_i-1.26 }{l_i (l_i-1)} \right)\right \rceil\nonumber \\ &\leq \sum_{l\in S_{\varphi}}  \left(\frac{2}{3}\right)^\varphi e^\omega \max\{1,\Phi_{\text{max}} ^{\omega}\} \prod_{l_i\geq C_K} \frac{1}{2^{l_i}l_i} \left( 2^{l_i} \left(\frac{0.76l_i-1.26 }{l_i (l_i-1)} \right)+1\right)\nonumber \\ & = \sum_{l\in S_{\varphi}}  \left(\frac{2}{3}\right)^\varphi e^\omega \max\{1,\Phi_{\text{max}} ^{\omega}\} \prod_{l_i\geq C_K}  \left(  \frac{0.76l_i-1.26 }{l_i^2 (l_i-1)} +\frac{1}{l_i 2^{l_i}}\right) \nonumber \\&\leq \sum_{l\in S_{\varphi}} \left(\frac{2}{3}\right)^\varphi e^\omega \max\{1,\Phi_{\text{max}} ^{\omega}\} \prod_{l_i\geq C_K}  \left(  \frac{0.76}{2} + \frac{1}{8}\right)\leq |S_\varphi| \left(\frac{2}{3}\right)^\varphi e^\omega \max\{1,\Phi_{\text{max}} ^{\omega}\}   
\end{align}
where for the first inequality we use $\binom{a}{b}\leq \left(\frac{a e}{b}\right)^b$. The second inequality uses the definition $\Phi_{\text{max}}:= \max_{j\in[1:C_K-1]} \frac{\pi(2^{j}-1)-\pi(2^{j-1}-1)}{j2^j}$ and inequality $(1/v_l(j))^{v_l(j)}\leq 1$ if $v_l(j)\neq 0$. Next, as $\sum_{j=1}^m v_l(j)= \omega$, it follows that
\begin{equation}
    \Phi_{\text{max}} ^{\sum_{j=1}^{C_K-1}v_l(j)}\leq \max\{1,\Phi_{\text{max}}^\omega \}
\end{equation}
For the fourth inequality is used that $l_i> 2$ and $0.76\,l_i-1.26\leq 0.76\,l_i$.
Finally, using that the number of integer solutions of $\sum_{i=1}^\omega l_i =\varphi$, with $l_i\geq 1$ is given by $\binom{\varphi-1}{\omega-1}$, then
\begin{equation}
    |S_\varphi|\leq \binom{\varphi-1}{\omega-1}
\end{equation}
Consequently, $|S_\varphi|=O(\varphi^{\omega-1})$, which implies that $A(\varphi)=O(\left(\frac{2}{3}\right)^{\varphi}\varphi^{\omega-1})$. Next, we obtain the lower bound of $A(\varphi)$. 
\begin{align}
    A(\varphi)\geq \sum_{l\in S_\varphi}&3^{-\varphi} \prod_{\substack{j=1\\v_l(j)\neq 0}}^{C_K-1} \frac{1}{j^{v_l(j)}}  \left(\frac{\pi(2^{j}-1)-\pi(2^{j-1}-1)}{v_l(j)}\right)^{v_l(j)} \prod_{\substack{j=C_K\\v_l(j)\neq 0}}^m \frac{1}{j^{v_l(j)}} \left(\frac{\left \lceil 2^j \left(\frac{0.76j-1.26}{j (j-1)} \right)\right \rceil}{v_l(j)}\right)^{v_l(j)}
    \nonumber\\&\geq \sum_{l\in S_\varphi} \left(\frac{2}{3}\right)^{\varphi}  \frac{\Phi_{\text{min}} ^{\sum_{j=1}^{C_K-1}v_l(j)}}{\omega^\omega} \prod_{l_i\geq C_K} \frac{0.76l_i-1.26}{l_i^2(l_i-1)}\nonumber\\&\geq \sum_{l\in S_\varphi} \left(\frac{2}{3}\right)^{\varphi}   \frac{\Phi_{\text{min}}^\omega}{\omega^\omega} \left(\frac{\omega}{\varphi}\right)^{2\omega}0.55^{\omega} = \frac{|S_\varphi|}{\varphi^{2\omega}}\left(\frac{2}{3}\right)^{\varphi}   \left(0.55\,\Phi_{\text{min}}\omega\right)^{\omega}
    \nonumber \\ & \geq \frac{1}{\varphi^{2\omega}}\left(\frac{2}{3}\right)^{\varphi}   \left(0.55\,\Phi_{\text{min}}\omega\right)^{\omega}
\end{align}
where the second inequality follows from the definition $\Phi_{\text{min}}:= \min_{j\in[1:C_K-1]} \frac{\pi(2^{j}-1)-\pi(2^{j-1}-1)}{j2^j}$, $v_l(j)\leq \omega$, and the inequality $\lceil u \rceil \geq u$. The third inequality uses
\begin{equation}
     \prod_{l_i\geq C_K} \frac{0.76l_i-1.26}{l_i^2(l_i-1)}\geq \prod_{l_i\geq C_K} \left(0.76-\frac{1.26}{l_i}\right) \prod_{i=1}^\omega \frac{1}{l_i^2},
\end{equation}
and the inequality between the geometric mean and the arithmetic mean,
\begin{equation}
    \left( \prod_{j=1}^\omega l_j \right)^{\frac{1}{\omega}}\leq \frac{1}{\omega} \sum_{j=1}^\omega l_j = \frac{\varphi}{\omega}
\end{equation}

\subsubsection{Step 3: Multiplicity of the prime factors}

The last step generates the multiplicity of the different prime factors, i.e., vector $r=[r_1,r_2,\cdots, r_\omega]^T$. In particular, each vector $r$ that satisfies 
\begin{equation}\label{cod_1}
    \prod_{i=1}^{\omega} p_i^{r_i} < 2^m
\end{equation}
and $r_i\geq1$ for all $i\in[1:\omega]$ is equally probable. To sample from this distribution, note that $r_i\leq m$ for all $i\in [1,\omega]$, i.e., there are at most $m^\omega$ values for the vector $r$. Therefore, we can generate in polynomial time a list of all vectors that satisfy \eqref{cod_1} by checking all vectors $r$ such that $1 \leq r_i\leq m$ for all $i\in [1:\omega]$.

Finally, the number of elements in the list is computed, and we decide $r$ by sampling from a uniform distribution over $[1:N_T]$, where $N_T$ is the number of elements in the list. Note that $N_T=|N_{p_1,\cdots,p_\omega}(2^m)|$.

\subsubsection{Probability Distribution}

In a nutshell, the proposed algorithm runs in polynomial time in $m$, and the distribution satisfies 
\begin{align}
    \frac{\text{Prob}_{X\sim D_m}(x)}{\text{Prob}_{X\sim U_m}(x)}\geq  \left(\frac{9}{11} \right)^2 \frac{\Xi(K)^{K^2}}{2^K}
\end{align}
where $U_m$ denotes the uniform distribution over $\{x\in \mathbb{N}: b(x)\leq m \text{ and }\omega(x)\leq K\}$. Therefore, $D_m$ satisfies expression $\eqref{eq:prop_suff_inf}$. 

The previous inequality follows directly from inequalities \eqref{inequality_1}, \eqref{inequality_2}, \eqref{inequality_3}, \eqref{distribution_equation}, and
\begin{equation}
    \mathrm{Prob}_{X\sim U_m} (x) = \left( \frac{\pi_\omega(2^m)}{\sum_{\omega=1}^K \pi_\omega(2^m)} \right) \left(\frac{|N_{p_1,\cdots,p_\omega}(2^m)|}{\pi_{\omega}(2^m)}\right) \left(\frac{1}{|N_{p_1,\cdots,p_\omega}(2^m)|}\right)
\end{equation}

\end{appendices}
\end{document}